    \documentclass[10pt,journal,compsoc]{IEEEtran}

    \usepackage[nocompress]{cite}
    \usepackage{amsmath,amsfonts}
    \usepackage{algorithmic}
    \usepackage{algorithm}
    \usepackage{array}
    \usepackage{booktabs}
    \usepackage[caption=false,font=normalsize,labelfont=sf,textfont=sf]{subfig}
    \usepackage{textcomp}
    \usepackage{stfloats}
    \usepackage{url}
    \usepackage{verbatim}
    \usepackage{graphicx}
    \usepackage{enumitem}
    \usepackage{amssymb}

    \usepackage{tikz}
    \usetikzlibrary{positioning, fit, calc, shapes.geometric, arrows.meta, backgrounds, patterns}

    \newtheorem{lemma}{Lemma}[section]
    \newtheorem{corollary}{Corollary}[section]
    \newtheorem{remark}{Remark}[section]

    \ifCLASSOPTIONcompsoc
    \else
    \usepackage{cite}
    \fi

    \ifCLASSINFOpdf
    \else
    \fi

\begin{document}
    %
    % Paper title
    \title{LipCache: A Local Inference Proxy with Certified Caching for Edge Image Classification Service}
    \author{Zhengzhe~Xiang,
        Yinlin~Chen,
        Fuli~Ying,
        Binbin~Zhou,
        Hailiang~Zhao~\IEEEmembership{Member,~IEEE},
	    Schahram~Dustdar~\IEEEmembership{Fellow,~IEEE}
        \thanks{This work was partially supported by the National Key R\&D Program of China (No.~2025YFG0100700) and the Supercomputing Center of Hangzhou City University. Additional support was provided by Grant CNS2023-144359, funded by MICIU/AEI/10.130391101100011033, and by the European Union NextGeneration EU/PRTR.}
        
        \IEEEcompsocitemizethanks{
            \IEEEcompsocthanksitem Z.~Xiang, Y.~Chen, F.~Ying, B.~Zhou are with Hangzhou City University, Hangzhou, China (e-mail: \{xiangzz@hzcu, chenyinlin@stu.hzcu, yingfuli@stu.hzcu, bbzhou@hzcu\}.edu.cn).
            \IEEEcompsocthanksitem H.~Zhao is with Zhejiang University, Hangzhou, China (e-mail:hliangzhao@zju.edu.cn)
            \IEEEcompsocthanksitem S.~Dustdar is with ICREA, Barcelona, Spain, and the Distributed Systems Group, TU Wien, Vienna, Austria (e-mail: dustdar@dsg.tuwien.ac.at).
            \IEEEcompsocthanksitem Copyright (c) 2026 IEEE. Personal use of this material is permitted. Permission from IEEE must be obtained for any other use, by sending a request to pubs-permissions@ieee.org.
        }
}

    \IEEEtitleabstractindextext{%
    \begin{abstract}
    As edge-side vision services continue to expand toward low-latency, high-throughput scenarios, reducing the inference cost of vision models without sacrificing reliability has become a central concern. Existing semantic caching methods largely rely on empirical similarity thresholds; while such thresholds improve hit rates, they tend to introduce silent misclassifications near decision boundaries. To address this, we propose \texttt{LipCache}, a certified semantic caching framework for image classification. Without modifying the existing deployed main model, \texttt{MainNet}, the framework introduces a lightweight network, \texttt{GuardNet}, that maps inputs into a low-dimensional feature space subject to a Lipschitz constraint. It then computes a per-sample certified reuse radius from the local classification margin and the spectral norm of the classification head. At runtime, a cached result is reused only when the query feature falls inside the certified reuse ball; otherwise, the query falls back to \texttt{MainNet}. Thus, cache hits are transformed from empirical threshold tests into geometric certification decisions with explicit theoretical boundaries. Across standard image classification tasks like CIFAR, Tiny-ImageNet, and SVHN, \texttt{LipCache} achieves a measured speedup of up to $1.65\times$ with limited end-to-end accuracy degradation, while all accepted cache hits satisfy the \texttt{GuardNet}-side certified-consistency condition. Furthermore, an enhanced \texttt{GuardNet} training recipe substantially improves cache hit rates in the Tiny-ImageNet multi-class extension while maintaining a certified-consistency rate of $100\%$. These results demonstrate that per-sample certified reuse can reduce main-model fallback while preserving theoretical consistency, providing a feasible approach to reliable cache-assisted inference at the edge.
\end{abstract}

    % Note: Peer-reviewed papers typically do not use keywords.
    \begin{IEEEkeywords}
    Edge intelligence, image classification, semantic caching
    \end{IEEEkeywords}
}
    \maketitle

    \IEEEdisplaynontitleabstractindextext
    \IEEEpeerreviewmaketitle

    \IEEEraisesectionheading{\section{Introduction}}
    \label{sec:introduction}

    \IEEEPARstart{D}{eep}-neural-network-based vision services have been widely deployed across diverse domains such as autonomous-driving perception, intelligent manufacturing inspection, real-time surveillance, and mobile augmented reality. The dual demand of these services for inference accuracy and response speed has made high-accuracy neural models a de facto standard component. However, the parameter count and computational complexity of such models keep growing, so that a single inference already consumes considerable compute and memory resources.

    This tension is further amplified in industrial edge scenarios. Unlike occasional photo recognition on a phone, industrial edge devices---such as production-line cameras and roadside perception units---typically operate continuously at high sampling rates, injecting a large volume of concurrent queries into the inference service within very short intervals. When these queries converge at the cloud, they readily trigger spikes in queries-per-second, network congestion, and server overload, severely threatening the overall service availability and tail latency. In other words, the root cause of the problem lies not only in the high cost of a single inference, but more critically in the fact that a large number of semantically redundant queries are repeatedly executed without discrimination.

    The rise of the edge computing paradigm offers a structural opportunity to mitigate this risk. By pushing part of the computation or storage capability to edge nodes close to the data source, the system can potentially absorb a fraction of the query load locally, thereby reducing its dependence on the remote cloud link. New challenges arise, however: the latency and memory budgets of edge hardware are extremely tight and far from sufficient to host a full-sized model comparable to its cloud counterpart \cite{xu2018scaling,shuvo2023edgeaccel,liu2024lightweight,wang2025edgeintelligence}. Directly deploying a high-accuracy classifier on an edge device almost inevitably faces a sharp conflict between accuracy and resources.

    To address this tension, existing work proceeds along roughly two lines. The first line aims to reduce the execution cost of the model at the edge, including compact architecture design, model compression, and early-exit strategies \cite{mobilenetv2_2018,shufflenet_2018,branchynet_2016}. These methods substantially improve edge deployability, yet their unit of optimization remains ``how to run the model more cheaply''; they do not touch upon a more fundamental question: for the many semantically repetitive queries, is it really necessary to run the model every time?

    The \textit{cache-assisted inference }starts precisely from this gap. It exploits the pervasive redundancy of visual patterns in edge workloads and reuses previously computed results to skip part of the model invocations directly \cite{DBLP:conf/acsos/OgdenGWG21,noscope2017,deepcache2018,li2021semanticmemory,deltacnn2022}. However, current reuse strategies suffer from a structural weakness in reliability: exact matching is extremely sensitive to compression, cropping, and illumination changes, and can hardly enable effective reuse in real deployments; semantic caching, although it enlarges the reuse region, mostly relies on heuristic similarity thresholds and provides no formal guarantee on the correctness of the reused labels. Consequently, cache hits near decision boundaries may introduce silent misclassifications and uncontrollable accuracy loss. A caching strategy that combines an explainable safety boundary with practical reuse efficiency is still missing.

    This paper proposes \texttt{LipCache}, a certified semantic caching framework that comprehensively exploits the storage and computation capability of edge devices. Its core design principle is modularity: rather than modifying the deployed high-accuracy main model (\texttt{MainNet}), it introduces a lightweight guard network, \texttt{GuardNet}, as a new branch to construct the semantic cache. Through operator-norm-controlled convolutions and spectral normalization, \texttt{GuardNet} maps inputs to a smooth low-dimensional feature space \cite{sedghi2019singular,miyato2018spectral}, and, starting from the local classification margin and the Lipschitz bound of the classification head, computes a per-sample \textit{certified reuse radius} offline for each cached sample. Online, only when the query feature falls inside the certified ball of some cached sample, namely when their distance is less than the certified reuse radius, does the system reuse the stored label; otherwise it falls back to \texttt{MainNet} to preserve accuracy. Thereby, a cache hit is converted from an empirical threshold decision into a geometric certification decision with an explicit theoretical boundary.

    We validate the core certification mechanism of \texttt{LipCache} on benchmarks with markedly different statistical structure: CIFAR (natural images), Tiny-ImageNet (higher-resolution fine-grained inter-class differences), and SVHN (domain-shifted digits and characters). This choice is designed to isolate the stability of the certified radius under different task geometries, rather than to exhaust all visual task scales; larger-scale settings (e.g., full ImageNet) and non-i.i.d. streaming scenarios are left as open directions discussed in Section~\ref{sec:conclusion}. On these three tasks, the average hit rates under the tight certified radius reach $0.342$, $0.472$, and $0.502$, respectively, with end-to-end accuracies of $0.921$, $0.867$, and $0.960$, and corresponding measured speedups of $1.32\times$, $1.31\times$, and $1.65\times$; moreover, the certified-consistency rate of all results is $100\%$. Comparisons with empirical thresholds and the global nearest-neighbor baseline show that the certified radius is the only reuse boundary that can stably preserve certified consistency across the three tasks. Furthermore, in the Tiny-ImageNet multi-class extension, an enhanced \texttt{GuardNet} training recipe improves the hit rates of the $20/30/50$-class tasks from $0.056/0.005/0.0004$ to $0.423/0.254/0.124$, while both the end-to-end accuracy and the certified consistency are preserved. The main contributions are as follows:
    \begin{enumerate}
        \item We propose a dual-model cache-assisted inference framework that comprehensively exploits the storage and computation capability of edge devices and, without modifying \texttt{MainNet}, makes certified reuse decisions through a lightweight \texttt{GuardNet}.
        \item Starting from the local classification margin and the Lipschitz bound of the \texttt{GuardNet} classification head, we derive a per-sample certified reuse radius in the feature space, providing a computable and interpretable rule for safe semantic reuse.
        \item We systematically evaluate \texttt{LipCache} across multiple benchmark tasks, multi-class extensions, and edge-deployment boundaries, verifying the effectiveness of the certified radius as the only cross-task stable certified boundary and revealing the potential of the enhanced training recipe in enlarging the certified reuse region.
    \end{enumerate}

    \section{Related Work}
    \label{sec:related_work}

    Edge image classification faces a structural tension: high-accuracy neural models keep growing in their computational and memory demands, whereas edge devices must satisfy tight latency, energy, storage, and connectivity constraints. A practical edge system therefore requires more than just a smaller classifier. It must reduce computation, exploit redundancy among correlated inputs, and decide when a low-cost shortcut is safe rather than merely fast. The literature addresses these needs along three main lines and one important bridge. Efficient edge inference reduces the cost of running a model through compact architectures, adaptive execution, hardware--software co-design, or edge--cloud partitioning. Cache-assisted inference and semantic reuse avoid part of the execution by reusing results from temporally or semantically correlated inputs. Certified nearest-neighbor retrieval forms a bridge, studying when a nearest-neighbor decision in an embedding space remains invariant under perturbation. Lipschitz control and certification provide the tools that turn margins and continuity bounds into local stability guarantees. \texttt{LipCache} connects these lines: it uses a lightweight \texttt{GuardNet} as a cache-oriented front end to the high-accuracy \texttt{MainNet}, and accepts reuse only inside a per-sample certified reuse radius in the \texttt{GuardNet} feature space.

    \subsection{Inference for Edge Vision Systems}
    \label{subsec:edge_vision}

    The dominant strategy for edge inference is to reduce the cost per invocation, or to place computation where it can be executed more efficiently. Survey work points out the underlying pressure: the scale and computational complexity of DNNs grow faster than the performance and efficiency gains that embedded platforms can provide \cite{xu2018scaling}. Recent surveys organize the response around efficient architectures, model compression, hardware acceleration, algorithm--hardware co-design, and deployment techniques for on-device or edge--cloud execution \cite{shuvo2023edgeaccel,liu2024lightweight,wang2025edgeintelligence,sun2025aiedgedeploy}. Edge--cloud collaboration further partitions a DNN so that shallow feature extraction runs on the device while deeper layers run on the cloud or an edge server, additionally attending to privacy, communication overhead, and resource allocation \cite{zhang2025edgecloudsurvey,fang2026mae}.

    Representative architecture-level methods reduce the cost of the backbone itself. MobileNetV2 uses inverted residual blocks, thin linear bottlenecks, and depthwise separable convolutions to improve mobile classification, detection, and segmentation across model sizes \cite{mobilenetv2_2018}. MobileNetV3 combines hardware-aware neural architecture search with NetAdapt and manual architectural refinements, yielding both large and small mobile models tuned for phone CPUs \cite{mobilenetv3_2019}. ShuffleNet targets extremely low computational budgets by combining pointwise grouped convolutions with channel shuffling, lowering cost while preserving mobile accuracy \cite{shufflenet_2018}. Representative execution-level methods adaptively reduce the average cost: BranchyNet adds side-branch classifiers so that high-confidence samples can exit early, while harder samples continue through deeper layers \cite{branchynet_2016}; MSDNet treats test-time computation as an anytime or budgeted classification problem and reuses multi-scale dense features across multiple early exits \cite{msdnet_2018}. At the collaborative-system level, MAE reduces intermediate-transmission and scheduling costs by activating sparse channel-level experts during DNN inference on edge devices \cite{fang2026mae}, and more recent edge-cloud co-inference combines spike-driven compression with dynamic early-exit to cut both transmission volume and end-to-end latency \cite{DBLP:conf/icmla/HassanDZZA25}.

    These works substantially improve deployability, yet their unit of optimization remains executing the current query---through the model, an early-exit path, or a device--cloud partition. They do not explicitly exploit a per-sample certified reuse region around cached samples to decide whether the current query can entirely skip the main classifier. \texttt{LipCache} is therefore complementary: it keeps \texttt{MainNet} unchanged as a reliable fallback predictor and adds an independent \texttt{GuardNet}-guided cache decision layer to avoid selected \texttt{MainNet} invocations.

    \subsection{Cache-Assisted Neural Network Inference}
    \label{subsec:cache_assisted}

    Cache-assisted inference follows a different strategy: rather than reducing the cost per invocation, it tries to avoid invocations by reusing outputs or intermediate states of correlated inputs. Exact-match memoization is too brittle for vision tasks, since repeated content may differ due to camera motion, illumination, scale, cropping, or semantic changes. Practical systems therefore extend reuse beyond equality by exploiting temporal locality, video-frame coherence, low-cost filters, intermediate feature maps, or semantic memory.

    Shadow Puppets proposes a semantic caching service that breaks the binary choice between pure cloud inference and full edge execution, leveraging edge-side caching to approximate cloud-level accuracy at lower latency and cost \cite{shadowpuppets2018}. DeepCache targets continuous mobile vision: it uses heuristics inspired by video compression to detect reusable regions in the video input, propagates these reusable regions through the layers of the CNN, works with unmodified models, and reports reductions in inference time and energy \cite{deepcache2018}. NoScope specializes video analytics to a specific video and object, searching over a cascade of cheap specialized models and difference detectors to mimic a reference network at a chosen accuracy level \cite{noscope2017}. DeltaCNN likewise exploits video coherence but at the implementation level: it propagates sparse inter-frame differences through typical CNN layers and accelerates per-frame inference without retraining \cite{deltacnn2022}. Semantic Memory turns reuse toward semantic features, encoding high-dimensional feature maps into low-dimensional semantic vectors, and uses hierarchical memory and adaptive cache management to accelerate mobile CNN inference with acceptable accuracy loss \cite{li2021semanticmemory}. More recent work continues to expand cache-assisted inference along several axes: synergistic lazy loading and layer-wise caching reduce the cold-start and loading cost of serverless edge models \cite{DBLP:conf/infocom/Huang0GY25}; task-aware expert caching decomposes multitask models into modular experts for fine-grained reuse \cite{DBLP:journals/iotj/SinthiaMSH25}; and distributed inference caching has been formalized as a submodular maximization problem under storage constraints \cite{DBLP:journals/tmc/ChenCH26a}. The closest line on certified retrieval is RetrievalGuard, which studies provably robust 1-nearest-neighbor image retrieval. Given a base retrieval model and a query, RetrievalGuard constructs a smoothed retrieval model, analyzes the 1-NN search process in the high-dimensional embedding space, and derives a computable $\ell_2$ radius within which the Recall@1 result remains unchanged \cite{wu2022retrievalguard}.

    These systems demonstrate that reuse can be a powerful source of acceleration, yet their acceptance criteria are mainly empirical or workload-specific. Video-based methods rely on temporal continuity or sparse inter-frame differences; specialized cascades rely on a fixed camera/query distribution; semantic memory relies on learned or measured feature similarity and exit policies; certified retrieval certifies the stability of the nearest-neighbor result itself, rather than the safe reuse of a cached payload in an edge inference pipeline. \texttt{LipCache} retains the idea of reuse but changes the acceptance criterion: the cache is built in a dedicated \texttt{GuardNet} feature space, and a reuse is accepted only when the query falls inside the per-sample certified reuse radius of some cached sample, a radius derived from the \texttt{GuardNet}-side margin and Lipschitz bounds.

    \subsection{Lipschitz Control and Certification}
    \label{subsec:lipschitz_control}

    Lipschitz control and certification address the safety problem of efficient inference. Their common strategy is to bound how fast a neural mapping can change, and then use this bound, together with margins or probabilistic arguments, to derive a region within which the prediction remains stable. This body of literature therefore provides the mathematical tools for turning similarity into local consistency statements.

    At the layer and architecture level, spectral normalization constrains the operator norm of the weights through a lightweight normalization procedure, originally introduced to stabilize GAN discriminators \cite{miyato2018spectral}. For convolutional layers, Sedghi et al.\ characterize the singular values of standard multi-channel convolutions, enabling efficient computation and projection onto operator-norm balls \cite{sedghi2019singular}. Araujo et al.\ provide a unified algebraic view of 1-Lipschitz neural networks, showing that several layers based on orthogonality, spectra, and convex potentials can be understood through a common semidefinite-programming condition \cite{DBLP:conf/iclr/AraujoHDA023}. These works explain how to build or constrain networks whose continuity can be controlled.

    At the certification level, CROWN bounds general activation functions with linear or quadratic surrogates to produce certified lower bounds on adversarial distortion \cite{crown2018}; randomized smoothing, in turn, turns a classifier that remains accurate under Gaussian noise into an $L_2$-certified smoothed classifier, extending certification to full-resolution ImageNet \cite{cohen2019smoothing}. More recent work tightens or extends Lipschitz-based certification. Huang et al.\ exploit activation states to remove inactive rows and columns from the induced-norm computation, thereby computing efficient local Lipschitz upper bounds \cite{DBLP:conf/nips/HuangZSKA21}. Fazlyab et al.\ use directional Lipschitz bounds and improved regularization to maximize a differentiable lower bound on the distance to the decision boundary \cite{fazlyab2023dynamicmargin}. Hu et al.\ show that achieving certifiable robustness under a Lipschitz constraint requires careful capacity and data design, including larger Lipschitz-controlled residual dense layers and data augmentation \cite{hu2024recipe}. Wang et al.\ improve scalability by reformulating semidefinite-programming Lipschitz estimation as an eigenvalue optimization problem \cite{wang2024lipscalability}; and ECLipsE decomposes the large verification problem into layer-scale subproblems or closed-form relaxations to obtain faster compositional estimates \cite{xu2024eclipse}.

    From the perspective of cache-assisted edge inference, the common limitation of these methods is that they certify model robustness or stability, rather than cache reuse. Their certificates typically answer whether a classifier's prediction remains invariant under perturbation; they do not decide whether a cached payload should replace an expensive \texttt{MainNet} invocation. \texttt{LipCache} repurposes the same mathematical ingredients for a system-level reuse decision: the \texttt{GuardNet} margin and the \texttt{GuardNet}-side Lipschitz control define a per-sample certified reuse radius in the \texttt{GuardNet} feature space. The resulting guarantee is local to \texttt{GuardNet}, whereas consistency with \texttt{MainNet} is treated as an empirical system property.

    In summary, existing work leaves three gaps with respect to reliable edge caching. Edge-inference methods reduce the cost of running a model but do not address whether the model can be skipped. Semantic caching and video reuse exploit redundancy but typically rely on empirical thresholds or workload locality. Certified-retrieval and robustness methods provide stability guarantees, but they do not formulate the cache-hit acceptance criterion as a \texttt{GuardNet}-side reuse certificate coupled with \texttt{MainNet} fallback. \texttt{LipCache} fills exactly this intersection by combining a dedicated \texttt{GuardNet} feature space, a per-sample certified reuse radius for each cached sample, and an empirical evaluation of the consistency between the accepted cache hits and \texttt{MainNet}.

    \section{Problem Description}
    \label{sec:problem}

    \subsection{System Setting and Cache Motivation}
    \label{subsec:task_definition}

    This section provides a formal model of cache-assisted inference in an edge image classification system, establishing a unified notational foundation for the framework design in Section~\ref{sec:methodology} and the certification analysis in Section~\ref{sec:theory}.

    Consider an image classification system whose input space is $\mathcal{X} \subseteq \mathbb{R}^{n}$ and whose label set is $\mathcal{Y}=\{1,\dots,C\}$. The system may invoke a high-accuracy classifier $M:\mathcal{X}\rightarrow\mathcal{Y}$, but a single inference of $M$ is costly---in edge scenarios this may amount to a combination of GPU time, cloud communication latency, and other overhead. For any input $x\in\mathcal{X}$, the system ultimately needs to output a prediction $\hat{y}\in\mathcal{Y}$.

    A key observation is that edge vision workloads are not composed of mutually independent random samples. The input--label pairs $(x,y)$ are drawn from a joint distribution $\mathcal{D}$ that often exhibits pronounced structural redundancy: influenced by factors such as a fixed camera viewpoint, periodic sampling, and slow scene changes, consecutively arriving queries are highly clustered in pixel space or semantic space, and adjacent inputs frequently share the same label. Under these conditions, executing $M$ in full for every query means that a large amount of computation is wasted on samples whose answers could have been inferred from existing results. This observation points to a natural optimization path: cache a subset of historical predictions, and when a new query is sufficiently similar to some cached sample, directly reuse its label, thereby skipping the invocation of $M$.

    However, the above path faces a fundamental tension. If the reuse condition is too loose, queries near a decision boundary may be incorrectly matched to an inappropriate cached label, leading to silent accuracy loss; if the reuse condition is too strict, cache hits will be so rare that the savings in $M$ invocations become negligible. The core challenge is therefore not ``whether to introduce a cache,'' but how to define a reuse rule that simultaneously achieves a sufficient hit rate to yield meaningful system gains and constrains the resulting prediction errors to a controllable range.

    \subsection{Formalization of the Optimization Objective}
    \label{subsec:optimization}

    To precisely characterize the above trade-off, we abstract cache-assisted inference as an optimization problem over deployment policies.

    A deployment policy $\pi$ produces two outputs for each input $x\in\mathcal{X}$: a predicted label $\hat{y}(x)\in\mathcal{Y}$, and a binary routing decision $\rho(x)\in\{\mathrm{cache},\,\mathrm{invoke}\}$, the latter indicating whether the prediction comes from cache reuse or from a direct inference of $M$.

    Let $(x,y)\sim\mathcal{D}$ denote a test sample drawn from the data distribution. The system accuracy of policy $\pi$ is defined as
    \begin{equation}
    \label{eq:acc_sys}
    \mathrm{Acc}_{\mathrm{sys}}(\pi)
    \triangleq \Pr_{(x,y)\sim\mathcal{D}}\!\bigl[\hat{y}(x)=y\bigr],
    \end{equation}
    and the standalone accuracy of $M$ is
    $\mathrm{Acc}_{M}
    \triangleq \Pr_{(x,y)\sim\mathcal{D}}\!\bigl[M(x)=y\bigr]$.
    In practice, the above distributional expectations are approximated by empirical frequencies on a reference validation or test set.

    The invocation rate of policy $\pi$ on $M$ is defined as
    \begin{equation}
    \label{eq:invoke}
    \mathrm{InvokeRate}_{M}(\pi)
    \triangleq \Pr_{(x,y)\sim\mathcal{D}}\!\bigl[\rho(x)=\mathrm{invoke}\bigr].
    \end{equation}
    This quantity directly determines the degree of cost savings on $M$ inference---the lower the invocation rate, the higher the fraction of $M$ invocations that are skipped.

    Cache-assisted inference can thus be formalized as the following constrained optimization problem:
    \begin{equation}
    \label{eq:objective}
    \begin{aligned}
    \min_{\pi}\quad & \mathrm{InvokeRate}_{M}(\pi) \\
    \mathrm{s.t.}\quad & \mathrm{Acc}_{\mathrm{sys}}(\pi)
     \ge \mathrm{Acc}_{M} - \Delta_{\mathrm{acc}},
    \end{aligned}
    \end{equation}
    where $\Delta_{\mathrm{acc}}\ge 0$ is a user-specified accuracy tolerance. A smaller $\Delta_{\mathrm{acc}}$ enforces a more conservative reuse strategy, tending to depress the hit rate to reduce accuracy risk; a larger $\Delta_{\mathrm{acc}}$ allows more aggressive cache reuse, trading potential accuracy for a lower $M$ invocation frequency.

    %This optimization is a deployment-level goal statement rather than a continuous optimization procedure to be solved directly. The path taken in this paper is to construct a parametric policy family defined jointly by the \texttt{GuardNet} parameters, the cache construction method, the radius mode, and the return semantics, then to bound the safe reuse region within this family through theoretical certification (Section~\ref{sec:theory}) and to compare the hit-rate--accuracy--latency trade-offs of different operating points through experiments (Section~\ref{sec:experiments}). Section~\ref{sec:methodology} gives the concrete instantiation of this policy family.

    \section{The LipCache}
    \label{sec:methodology}

    Section~\ref{sec:problem} formalizes cache-assisted inference as a deployment-level optimization problem that minimizes the main-model invocation rate under an accuracy-tolerance constraint. This section gives a concrete instantiation of that problem---the \texttt{LipCache} framework. Its workflow proceeds in three stages: first, a Lipschitz-constrained lightweight guard network, \texttt{GuardNet}, is trained offline (Section~\ref{subsec:guardnet_training}); next, a feature cache is built on the trained \texttt{GuardNet}, and each cached sample is assigned a per-sample certified reuse radius (Section~\ref{subsec:cache_construction}); finally, in the online phase, \texttt{GuardNet} guides cache-hit detection and the fallback decision (Section~\ref{subsec:online_inference}). We first establish the unified notation.

    \subsection{Setup and Notation}
    \label{subsec:setup}

    \texttt{LipCache} adopts a dual-model design, comprising a high-accuracy main model $M$ and a lightweight guard model $G$. The main model provides the reference prediction:
    \begin{equation}
    \hat{y}_{M}(x)=\arg\max_{c\in\mathcal{Y}} M_c(x).
    \end{equation}
    The guard model maps each input to a low-dimensional feature space,
    $G:\mathcal{X}\rightarrow\mathbb{R}^{d}$. Let $z=G(x)$; an affine
    classifier operates in this space:
    \begin{equation}
    s(z)=Wz+b,
    \end{equation}
    where $W\in\mathbb{R}^{C\times d}$, $b\in\mathbb{R}^{C}$, yielding the
    \texttt{GuardNet} prediction
    \begin{equation}
    \hat{y}_{G}(x)=\arg\max_{c\in\mathcal{Y}} s_c(G(x)).
    \end{equation}
    We denote the full logit map by $F_G(x)=s(G(x))$ and the affine head by
    $f(z)=Wz+b$. When two logit vectors are subtracted, the bias term cancels,
    so it does not affect the Lipschitz constant.

    For each cached sample, the certified reuse radius is derived from the local classification margin and the Lipschitz bound of $f$
    (Section~\ref{sec:theory}). At runtime, a query that falls inside the certified reuse region of some cached sample reuses the stored \emph{payload}
    (the label returned upon a cache hit);
    otherwise the system falls back to $M$. Table~\ref{tab:symbols}~summarizes the notation used.

    \begin{table}[!t]
    \centering
    \caption{Main symbols used in the formulation.}
    \label{tab:symbols}
    \begin{tabular}{ll}
    \hline
    Symbol & Meaning \\
    \hline
    $\mathcal{X}, \mathcal{Y}$ & Input space and label set \\
    $M$ & High-accuracy main model (\texttt{MainNet}) \\
    $G$ & Feature map of the guard model (\texttt{GuardNet}) \\
    $W$ & Weight matrix of the \texttt{GuardNet} affine classifier \\
    $d$ & Dimensionality of the \texttt{GuardNet} feature space \\
    $L_G$ & Lipschitz constant of the \texttt{GuardNet} feature map $G$ \\
    $\sigma_{\max}(W)$ & Spectral norm of $W$ \\
    $z_i$ & \texttt{GuardNet} feature of cached sample $x_i$ \\
    $y_i^G$ & \texttt{GuardNet} prediction of $x_i$ \\
    $\tilde{y}_i$ & Cached payload associated with $x_i$ \\
    $r_i$ & Certified reuse radius of cached feature $z_i$ \\
    $z_{\mathrm{new}}$ & \texttt{GuardNet} feature of a new input \\
    $m_i$ & \texttt{GuardNet} classification margin of $x_i$ \\
    $\mathcal{C}$ & Feature cache maintained by the system \\
    \hline
    \end{tabular}
    \end{table}

    Each cache entry is a triple $(z_i,\tilde{y}_i,r_i)$, containing a \texttt{GuardNet}
    feature, a payload label, and a certified reuse radius.
    Figure~\ref{fig:guardnet_kd}~shows the offline training and feature-constraint pipeline,
    and Figure~\ref{fig:inference_flow}~shows the online inference.

    \begin{figure*}[!htbp]
    \centering
    \includegraphics[width=0.95\linewidth]{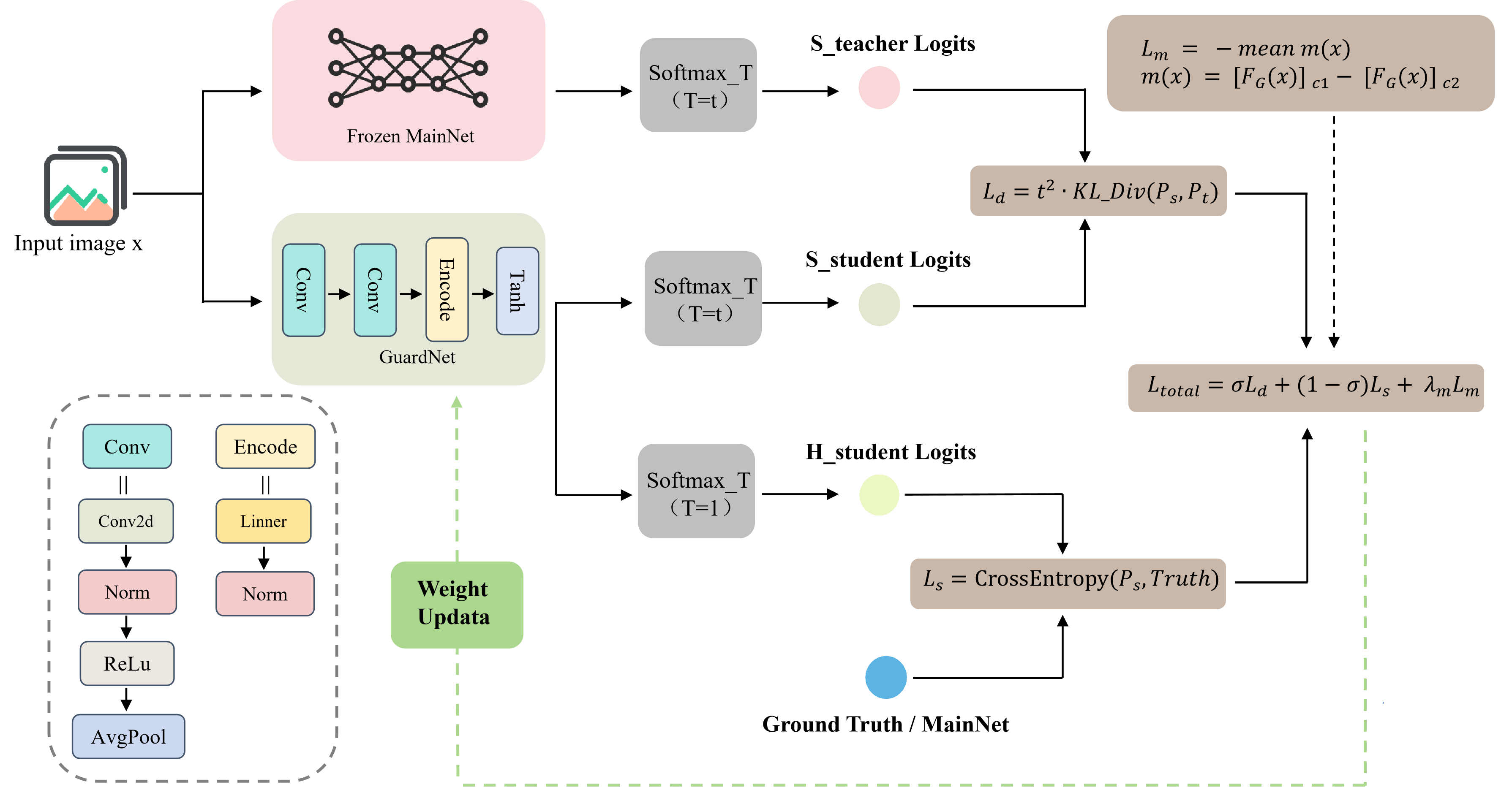}
    \caption{Offline training pipeline of \texttt{GuardNet}. Lipschitz control on the convolutions and spectral normalization on the linear layers are used to constrain the \texttt{GuardNet} feature map, thereby supporting the subsequent \texttt{GuardNet}-side certification; the main model may provide reference predictions when needed, but it is not a required training signal in the three 10-class main experiments of this paper.}
    \label{fig:guardnet_kd}
    \end{figure*}

    \subsection{Offline \texttt{GuardNet} Training}
    \label{subsec:guardnet_training}

    With the notation established, this subsection answers the first core question: how to train a \texttt{GuardNet} so that its learned feature geometry both supports discriminative classification and provides a sufficient certified reuse region for subsequent certified caching. To this end, the training objective is designed as a layered framework:
    \begin{equation}
    L_{\mathrm{total}} = (1-\alpha)L_{\mathrm{CE}} + \alpha L_{\mathrm{distill}}
     + \lambda_m L_{\mathrm{margin}},
    \end{equation}
    where $\alpha\in[0,1]$ balances direct supervision against an optional distillation signal from $M$, and $\lambda_m$ controls the margin-regularization strength. When $\alpha=0$, the objective reduces to the pure cross-entropy training used in the 10-class main experiments of this paper; when $\alpha>0$, it can be used for additional distillation ablations or teacher-alignment experiments. The above objective subsumes three complementary training signals, which we elaborate on one by one below, starting with the structural constraint that supports them.

    \textbf{Lipschitz control.}~
    The theoretical certificate in Section~\ref{sec:theory} requires the feature map to satisfy $\mathrm{Lip}(G)\le 1$. To this end, \texttt{GuardNet} imposes an operator-norm constraint on all parameterized layers. For a convolutional layer whose kernel is
    $K\!\in\!\mathbb{R}^{c_{\mathrm{out}}\times c_{\mathrm{in}}\times k_1\times k_2}$,
    the Lipschitz constant of the convolution operator is given by its frequency-domain representation:
    \begin{equation}
    \mathrm{Lip}(K)=\max_{\omega}\,\sigma_{\max}\!\big(\hat{K}(\omega)\big),
    \end{equation}
    where $\hat{K}(\omega)$ is the frequency-domain matrix obtained by zero-padding the kernel and applying a two-dimensional real FFT. Dividing the kernel by
    $\max(\mathrm{Lip}(K),1)$ before each forward pass enforces the 1-Lipschitz constraint~\cite{sedghi2019singular}. Linear layers use standard spectral normalization~\cite{miyato2018spectral}; ReLU, non-overlapping average pooling, and the final $\tanh$ projection are all non-expansive~\cite{DBLP:conf/iclr/AraujoHDA023}. Since the Lipschitz constant of each individual layer does not exceed 1, by submultiplicativity of operator norms under composition, the overall guard network satisfies $\mathrm{Lip}(G)\le 1$.

    \textbf{Distillation loss (optional).}~
    To optionally leverage the knowledge of $M$, one may minimize the KL divergence between the softened outputs of $M$ and $F_G$ at temperature $T$:
    \begin{equation}
    L_{\mathrm{distill}} = T^2\,\mathrm{KL}\!\bigl(
    \mathrm{Softmax}\!\bigl(\tfrac{M(x)}{T}\bigr)
    \,\big\|\,
    \mathrm{Softmax}\!\bigl(\tfrac{F_G(x)}{T}\bigr)\bigr).
    \end{equation}
    The standard cross-entropy $L_{\mathrm{CE}}$ with the ground-truth label $y$ preserves discriminative performance;
    the main experiments in this paper use $\alpha=0$, i.e., this term is disabled.

    \textbf{Margin regularization.}~
    The two terms above guarantee the discriminative accuracy of \texttt{GuardNet}, whereas margin regularization directly serves cache efficiency: the certified reuse radius $r_i$ is proportional to the \texttt{GuardNet} classification margin (see Section~\ref{sec:theory}), so explicitly encouraging large margins is a key lever for improving the hit rate. Let $c_1$ and $c_2$ denote the indices of the top-1 and top-2 logits:
    \begin{equation}
    m(x) = [F_G(x)]_{c_1} - [F_G(x)]_{c_2}.
    \end{equation}
    To avoid blindly enlarging the margin on already misclassified samples, the regularizer is applied only to the correctly classified subset of a batch $\mathcal{B}$:
    \begin{equation}
     \mathcal{S}_{\mathrm{corr}}
     = \left\{(x,y)\in\mathcal{B}\mid
     \arg\max_{c\in\mathcal{Y}} [F_G(x)]_c = y\right\},
    \end{equation}
    \begin{equation}
    L_{\mathrm{margin}}
    = - \frac{1}{\max(1,|\mathcal{S}_{\mathrm{corr}}|)}
    \sum_{(x,y)\in\mathcal{S}_{\mathrm{corr}}} m(x).
    \end{equation}
    This selectively enlarges the margins around samples that are already correctly predicted, and these margins directly determine the certified cache radius.

    \textbf{Enhanced training recipe (optional).}~
    For multi-class extensions at higher class counts, cross-entropy alone or a single margin term is often insufficient to sustain a usable certified radius. We therefore additionally consider an enhanced training recipe: beyond the core objective above, it further adds a center constraint that promotes intra-class compactness and an additional inter-class separation term that promotes inter-class separation, while optionally retaining the distillation signal. These additional terms only change the feature geometry learned by \texttt{GuardNet}; they do not change the online hit rule of the cache, nor do they change the certification criterion in Section~\ref{sec:theory}.

    In the main-text experiments, the three 10-class main results use cross-entropy supervision only ($\alpha=0,\lambda_m=0$); the CIFAR-10 training-objective ablation examines the independent effects of the margin and center constraints, respectively; and the Tiny-ImageNet multi-class extension uses the full enhanced recipe comprising the margin, interclass, center, and distill terms.

    \subsection{Feature Cache Construction}
    \label{subsec:cache_construction}

    Once \texttt{GuardNet} is trained, the framework enters the offline cache-construction phase. Its task is to select a set of high-quality cache entries from the reference set $\mathcal{D}_{\mathrm{ref}}$, so that it can cover the queries in the test distribution as much as possible under a given capacity budget.

    For each candidate sample $x_i$, we first compute $z_i=G(x_i)$, the \texttt{GuardNet} label
    $y_i^G=\arg\max_{c} [F_G(x_i)]_c$, and the payload $\tilde{y}_i$ selected by the construction strategy.
    This paper considers three payload semantics. The first is the \emph{proxy label},
    which directly stores $\tilde{y}_i = y_i^G$; in this case the cached payload is naturally aligned with the certificate, and it is the default setting adopted in the main-text experiments.
    The second is the \texttt{MainNet} prediction scheme, which stores $\tilde{y}_i = \hat{y}_M(x_i)$;
    the third is the \emph{ground-truth} scheme, which stores $\tilde{y}_i = y_i$.
    The latter two are mainly used to analyze the effect of varying the payload source on system behavior.

    The \texttt{GuardNet} margin $m_i=[F_G(x_i)]_{c_1}-[F_G(x_i)]_{c_2}$ is used to compute
    the certified radius $r_i$. Depending on the radius mode $\rho$, this paper supports two choices: the conservative radius
    $r_i=m_i/(\sigma_{\max}(W)\sqrt{2})$ (Lemma~\ref{lemma:safety}) or the tighter per-pair radius $r_i^*$ (Remark~\ref{remark:tighter_radius}). Candidates with non-positive radius are discarded outright.

    The certificate is attached to $y_i^G$. When the cache-center label consistency
    $y_i^G=\tilde{y}_i$ holds, the returned payload is directly aligned with the certificate;
    otherwise the reliability of the payload is assessed empirically.
    To maximize the coverage under a fixed budget, the selection strategy balances radius size against feature-space diversity.
    Algorithm~\ref{alg:lipcache_construction}~summarizes this construction process.

    \begin{algorithm}[!t]
    \caption{Offline construction of the feature store.}
    \label{alg:lipcache_construction}
    \begin{algorithmic}[1]
    \REQUIRE Reference set $\mathcal{D}_{\mathrm{ref}}$, models $M, G, W$, cache budget $B_C$
    \REQUIRE Payload construction rule, payload source (proxy, \texttt{MainNet} $M$ prediction, or label $y$), cache selection strategy $\mathcal{S}$
    \REQUIRE Radius mode $\rho$, optional cache-center label-consistency flag
    \ENSURE Feature cache $\mathcal{C}$
    \STATE Initialize candidate pool $\mathcal{P}\leftarrow\varnothing$
    \STATE Compute $\sigma_{\max}(W)$
    \IF{$\sigma_{\max}(W)\le 0$}
        \STATE Report that the \texttt{GuardNet} classification head is degenerate and stop construction
    \ENDIF
    \FOR{each candidate sample $x_i\in\mathcal{D}_{\mathrm{ref}}$}
        \STATE Compute $z_i=G(x_i)$ and \texttt{GuardNet} logits $F_G(x_i)$
        \STATE Compute \texttt{GuardNet} label $y_i^G$, payload $\tilde{y}_i$, and margin $m_i$
        \STATE Compute certified radius $r_i$ using radius mode $\rho$
        \IF{cache-center label consistency is required and $y_i^G\neq\tilde{y}_i$}
            \STATE Skip $x_i$
        \ELSIF{$r_i\le 0$}
            \STATE Skip $x_i$
        \ELSE
            \STATE Add candidate entry $(z_i,\tilde{y}_i,r_i)$ to $\mathcal{P}$
        \ENDIF
    \ENDFOR
    \STATE Use $\mathcal{S}$ to select at most $B_C$ entries from $\mathcal{P}$
    \STATE The selection process balances certified radius and feature diversity
    \STATE Store the selected entries in $\mathcal{C}$
    \RETURN $\mathcal{C}$
    \end{algorithmic}
    \end{algorithm}

    \subsection{Online Inference and Cache-Hit Detection}
    \label{subsec:online_inference}

    After offline training and cache construction are complete, \texttt{LipCache} enters the online inference state. At runtime (Figure~\ref{fig:inference_flow}), each input $x_{\mathrm{new}}$ is
    encoded as $z_{\mathrm{new}}=G(x_{\mathrm{new}})$, and the hit set is
    \begin{equation}
    \mathcal{H}(z_{\mathrm{new}})=\{i \mid \|z_{\mathrm{new}}-z_i\|_2 < r_i\}.
    \end{equation}
    If $\mathcal{H}\neq\varnothing$, it takes the payload from the nearest hit entry:
    \begin{equation}
    i^{*}(x_{\mathrm{new}})=\arg\min_{i\in\mathcal{H}(z_{\mathrm{new}})} \|z_{\mathrm{new}}-z_i\|_2.
    \end{equation}
    The system predicts
    \begin{equation}
    \hat{y}_{\mathrm{sys}}(x_{\mathrm{new}})=
    \begin{cases}
    \tilde{y}_{i^{*}(x_{\mathrm{new}})}, & \mathcal{H}(z_{\mathrm{new}})\neq\varnothing,\\[4pt]
    \hat{y}_{M}(x_{\mathrm{new}}), & \mathcal{H}(z_{\mathrm{new}})=\varnothing.
    \end{cases}
    \end{equation}
    On a miss, the system falls back to $M$; if write-back is enabled, a missed sample may be added to the cache under the same construction strategy.

    \begin{figure*}[!htbp]
    \centering
    \includegraphics[width=0.9\linewidth]{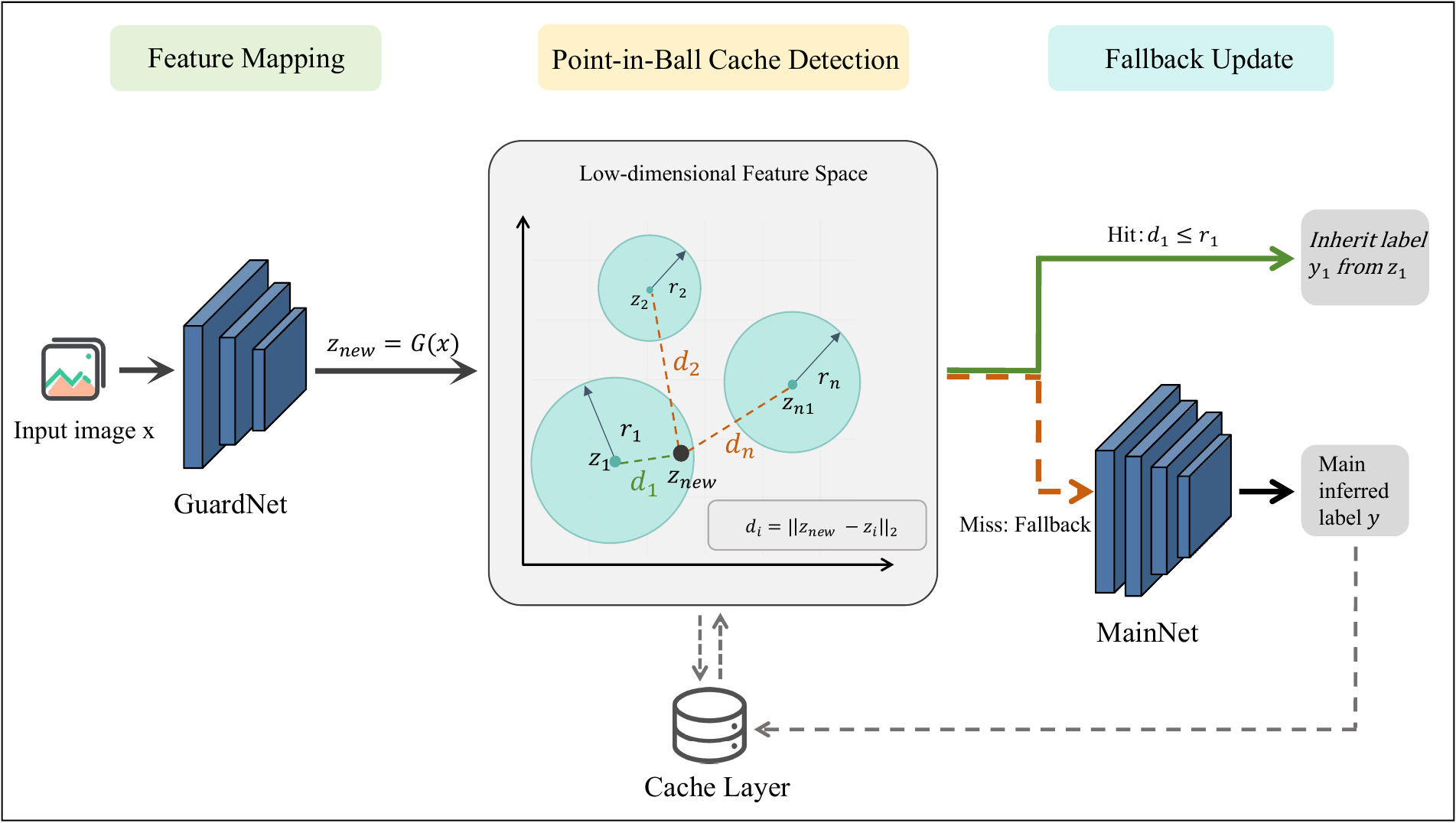}
    \caption{Online inference pipeline of \texttt{LipCache}. For a new input, \texttt{GuardNet} extracts a compact feature and uses it to search the cache for certified reuse. A cache hit returns the stored payload; a miss triggers a \texttt{MainNet} fallback.}
    \label{fig:inference_flow}
    \end{figure*}

    The computational overhead of the above pipeline directly determines the practical gain of the framework. Let $C_G$ and $C_M$ denote the per-query cost of \texttt{GuardNet} and \texttt{MainNet}, respectively. When maintaining $N$ cache entries in a
    $d$-dimensional feature space, the expected cost per query is
    \begin{equation}
    \mathbb{E}[\mathrm{Cost}]
    =O(C_G + Nd + \mathrm{InvokeRate}_{M}\,C_M).
    \end{equation}
    When $C_G+Nd$ is much smaller than $C_M$, the system can reduce the overall overhead in expectation even if the hit rate is not high. For large caches, hierarchical indexing or a vector database can further reduce the search complexity to sublinear, without changing the decision rule itself.
    Algorithm~\ref{alg:cache_hit}~summarizes the online pipeline.

    \begin{algorithm}[!t]
    \caption{Online inference with cache reuse.}
    \label{alg:cache_hit}
    \begin{algorithmic}[1]
    \REQUIRE Query batch $\{x_k\}_{k=1}^{B}$, main model $M$, trained \texttt{GuardNet}
    encoder $G$ and classification head $W$, feature cache
    $\mathcal{C}=\{(z_i,\tilde{y}_i,r_i)\}_{i=1}^{N}$, optional write-back flag
    \ENSURE Predictions $\{\hat{y}_{\mathrm{sys}}(x_k)\}_{k=1}^{B}$
    \STATE Compute \texttt{GuardNet} features $z_k \leftarrow G(x_k)$ for all $k$
    \STATE Compute pairwise distances $D_{k,i}=\|z_k-z_i\|_2$
    \FOR{$k=1$ to $B$}
        \STATE Construct hit set $\mathcal{H}(z_k)=\{i \mid D_{k,i}< r_i\}$
        \IF{$\mathcal{H}(z_k)\neq\varnothing$}
            \STATE $i_k^*=\arg\min_{i\in\mathcal{H}(z_k)} D_{k,i}$
            \STATE $\hat{y}_{\mathrm{sys}}(x_k)\leftarrow \tilde{y}_{i_k^*}$
        \ELSE
            \STATE $\hat{y}_{\mathrm{sys}}(x_k)\leftarrow \arg\max M(x_k)$
            \IF{write-back is enabled and the required payload source is available}
                \STATE Compute radius and payload for $x_k$ using \texttt{GuardNet} logits and $W$; update $\mathcal{C}$
            \ENDIF
        \ENDIF
    \ENDFOR
    \RETURN $\{\hat{y}_{\mathrm{sys}}(x_k)\}_{k=1}^{B}$
    \end{algorithmic}
    \end{algorithm}

    Combining the three stages above, \texttt{LipCache} instantiates the abstract deployment policy of Section~\ref{sec:problem} as a parametric policy family
    $\pi=(\theta_G,\mathcal{C},\rho,\mathcal{R},\omega)$,
    where $\theta_G$ is the \texttt{GuardNet} parameter set, $\mathcal{C}$ is the feature cache, $\rho$ is the radius mode, $\mathcal{R}$ is the hit-resolution rule, and $\omega$ controls the payload-source and cache-center label-consistency options. Different parameter combinations correspond to different operating points---these are precisely the core objects of the experimental evaluation in Section~\ref{sec:experiments}.

    \section{Theoretical Analysis}
    \label{sec:theory}

    Section~\ref{sec:methodology} gives \texttt{LipCache} a complete engineering pipeline, yet its core decision---the cache-hit test---is still an empirical definition: $\|z_{\mathrm{new}}-z_i\|_2<r_i$. This section injects theoretical content into that decision: starting from the Lipschitz property of the classifier in the \texttt{GuardNet} feature space and the local classification margin, we derive a per-sample certified reuse radius $r_i$ and prove that any query falling inside this radius necessarily preserves the same \texttt{GuardNet} prediction label as the cache center. Thereby, the accuracy constraint $\Delta_{\mathrm{acc}}$ in the optimization of Section~\ref{sec:problem} obtains a computable lower bound at the \texttt{GuardNet} level---as long as cache hits occur strictly within the certified radius, label consistency on the \texttt{GuardNet} side is guaranteed mathematically.

    \subsection{Lipschitz Setup}
    \label{subsec:lipschitz_setting}

    The structure of the certified reuse radius can be intuitively understood as ``classification margin divided by the Lipschitz bound'': the larger the margin, the larger the admissible perturbation; the smoother the mapping, the smaller the logit change induced by the same input variation. We therefore first establish the Lipschitz bound of the $\texttt{GuardNet}$ composite map $F_G=f\circ G$.

    All Lipschitz constants are measured under the $\ell_2$ norm. The feature encoder $G$ has a finite Lipschitz constant
    \begin{equation}
    L_G \triangleq \mathrm{Lip}(G),
    \end{equation}
    and the per-layer spectral-normalization construction of Section~\ref{subsec:guardnet_training} enforces $L_G\le 1$.
    For the affine classification head $f(z)=Wz+b$, noting that the bias term cancels automatically in a logit difference, we have
    \begin{equation}
    \|f(z_1)-f(z_2)\|_2
    = \|W(z_1-z_2)\|_2
    \le \sigma_{\max}(W)\|z_1-z_2\|_2,
    \end{equation}
    so $f$ is $\sigma_{\max}(W)$-Lipschitz. By submultiplicativity,
    \begin{equation}
    \mathrm{Lip}(F_G)
    = \mathrm{Lip}(f\circ G)
    \le L_G\,\sigma_{\max}(W).
    \end{equation}
    Under the constraint $L_G\le 1$, this simplifies to $\mathrm{Lip}(F_G)\le\sigma_{\max}(W)$. This means that along any direction in the feature space, the rate of change of $F_G$'s output is controlled by $\sigma_{\max}(W)$. All certificates below are stated in the \texttt{GuardNet} feature space and depend only on $\sigma_{\max}(W)$---this is also the reason why the cache-construction phase in Section~\ref{sec:methodology} needs to compute only this single spectral norm.

    \subsection{Certified Local Label Consistency}
    \label{subsec:certified_consistency}

    With the global Lipschitz bound in hand, we now turn to the local analysis: for a fixed cached sample, within what neighborhood does its \texttt{GuardNet} prediction remain unchanged.

    For a cached sample $x_i$, let its feature be $z_i=G(x_i)$ and its \texttt{GuardNet} prediction label be $y_i^G=\arg\max_{c} f_c(z_i)$. Define the local classification margin of this sample as
    \begin{equation}
    m_i=[f(z_i)]_{y_i^G}-\max_{j\neq y_i^G}[f(z_i)]_j.
    \end{equation}
    Geometrically, $m_i$ measures the logit-space margin between $z_i$ and the nearest decision boundary---the larger the margin, the farther the sample lies from a decision boundary, and the larger the perturbation required to flip its label. This margin-based certification idea is widely adopted in the certified-robustness literature~\cite{DBLP:conf/nips/HuangZSKA21,fazlyab2023dynamicmargin,hu2024recipe,wang2024lipscalability,xu2024eclipse}. Accordingly, the certified reuse radius is defined as
    \begin{equation}
    r_i=\frac{m_i}{\sigma_{\max}(W)\sqrt{2}},
    \end{equation}
    where $\sigma_{\max}(W)>0$; the degenerate case $\sigma_{\max}(W)=0$ means that the classification head degenerates to a constant map, in which case the reuse radius of any positive-margin sample is unbounded, but this does not occur in practice. The $\sqrt{2}$ factor in the denominator comes from the worst-case analysis that simultaneously bounds all competing classes, and it will be tightened in Remark~\ref{remark:tighter_radius}.

    The following lemma is the core theoretical result of this framework; it establishes the guarantee that ``the label is invariant inside the certified ball.''

    \begin{lemma}[Local label consistency]
    \label{lemma:safety}
    Let $z_i$ be a cached feature with \texttt{GuardNet} label $y_i^G$ and positive margin
    $m_i>0$, and assume $\sigma_{\max}(W)>0$. For any new sample
    $x_{\mathrm{new}}$ whose \texttt{GuardNet} feature is
    $z_{\mathrm{new}}=G(x_{\mathrm{new}})$, if
    \begin{equation}
    \|z_{\mathrm{new}}-z_i\|_2 < r_i = \frac{m_i}{\sigma_{\max}(W)\sqrt{2}},
    \end{equation}
    then the \texttt{GuardNet} prediction at $z_{\mathrm{new}}$ is still $y_i^G$.
    \end{lemma}

    The core idea of the proof is to use the Lipschitz bound to translate the displacement in feature space into an upper bound on the perturbation in logit space, and then to show that this perturbation is insufficient to cover the original margin $m_i$.

    \textit{Proof:} Let
    $\delta = f(z_{\mathrm{new}})-f(z_i)$. By the Lipschitz property of $f$,
    \begin{equation}
    \|\delta\|_2 \le \sigma_{\max}(W)\|z_{\mathrm{new}}-z_i\|_2.
    \end{equation}
    To preserve the label $y_i^G$, it suffices to show that for every competing class $j\neq y_i^G$ we have $[f(z_{\mathrm{new}})]_{y_i^G}>[f(z_{\mathrm{new}})]_{j}$. Substituting $f(z_{\mathrm{new}})=f(z_i)+\delta$ into the logit difference:
    \begin{equation}
    \begin{aligned}
    \relax [f(z_{\mathrm{new}})]_{y_i^G} - [f(z_{\mathrm{new}})]_j
    &= [f(z_i)]_{y_i^G}+\delta_{y_i^G} - [f(z_i)]_j-\delta_j \\
    = \big([f(z_i)]_{y_i^G}-&[f(z_i)]_j\big)
     - \big(\delta_j-\delta_{y_i^G}\big).
    \end{aligned}
    \end{equation}
    By the definition of the margin, $[f(z_i)]_{y_i^G}-[f(z_i)]_j \ge m_i$, so a sufficient condition for preserving the label is
    \begin{equation}
    \delta_j-\delta_{y_i^G}<m_i.
    \end{equation}
    To bound the left-hand side, introduce the auxiliary vector $e_{j,y_i^G}$: it takes the value $1$ at position $j$, $-1$ at position $y_i^G$, and $0$ elsewhere. Then $\|e_{j,y_i^G}\|_2 = \sqrt{2}$, and
    \begin{equation}
    \delta_j-\delta_{y_i^G}
    = e_{j,y_i^G}^{\top}\delta
    \le \|e_{j,y_i^G}\|_2\|\delta\|_2
    = \sqrt{2}\|\delta\|_2,
    \end{equation}
    where the inequality follows from Cauchy--Schwarz. Combining with the Lipschitz bound,
    \begin{equation}
    \delta_j-\delta_{y_i^G}
    \le \sqrt{2}\sigma_{\max}(W)\|z_{\mathrm{new}}-z_i\|_2.
    \end{equation}
    If $\|z_{\mathrm{new}}-z_i\|_2 < m_i/(\sigma_{\max}(W)\sqrt{2})$, then $\delta_j-\delta_{y_i^G} < m_i$, i.e., the sufficient condition holds. Substituting gives $[f(z_{\mathrm{new}})]_{y_i^G}-[f(z_{\mathrm{new}})]_j > 0$ for all $j\neq y_i^G$, and hence $y_i^G$ remains the argmax label at $z_{\mathrm{new}}$.\hfill$\blacksquare$

    This lemma directly provides the safety criterion for online hits in Section~\ref{sec:methodology}: as long as the query feature falls inside the certified ball, the \texttt{GuardNet} classification of that query cannot differ from the cache center---in other words, the cached payload $\tilde{y}_i$ (when set to the proxy label $y_i^G$) is theoretically reliable.

    \begin{corollary}[Input-level sufficient condition]
    \label{cor:input}
    Given $L_G\le 1$, any input satisfying $\|x_{\mathrm{new}}-x_i\|_2<r_i$ is guaranteed to satisfy $\|z_{\mathrm{new}}-z_i\|_2<r_i$, and hence its \texttt{GuardNet} prediction is $y_i^G$. More generally, if $L_G>0$, the sufficient condition is $\|x_{\mathrm{new}}-x_i\|_2<r_i/L_G$.
    \end{corollary}

    \textit{Proof:}\quad
    $\|z_{\mathrm{new}} - z_i\|_2
    = \|G(x_{\mathrm{new}}) - G(x_i)\|_2
    \le L_G \|x_{\mathrm{new}} - x_i\|_2
    < r_i$.
    The conclusion follows directly from Lemma~\ref{lemma:safety}.\hfill$\blacksquare$

    The above corollary shows that, when \texttt{GuardNet} is 1-Lipschitz, the $r_i$-ball in input space directly constitutes a sufficient region for certified reuse. The online cache decision, however, still uses the feature-space check directly---it is tighter and requires no back-propagation through $L_G$.

    \subsection{Tighter Radius and Certificate Scope}
    \label{subsec:tighter_radius}

    \begin{remark}[Tighter per-pair radius]
    \label{remark:tighter_radius}
    The $\sqrt{2}$ factor in the radius $r_i$ of Lemma~\ref{lemma:safety} comes from taking the worst case over all $C-1$ competing classes simultaneously. By analyzing each class separately, a tighter radius can be obtained. Let $W_c\!\in\!\mathbb{R}^{d}$ denote the $c$-th row of $W$. For each competing class $j\neq y_i^G$, decompose the logit difference as
    \begin{multline}
    [f(z)]_{y_i^G} - [f(z)]_j
    = \big([f(z_i)]_{y_i^G} - [f(z_i)]_j\big)\\
    + (W_{y_i^G} - W_j)^{\!\top}(z - z_i).
    \end{multline}
    By Cauchy--Schwarz, the cross term is bounded by $-\|W_{y_i^G} - W_j\|_2\,\|z - z_i\|_2$, so when $\|z - z_i\|_2 < r_{i,j}$ we have $[f(z)]_{y_i^G} > [f(z)]_j$, where
    \begin{equation}
    r_{i,j} = \frac{[f(z_i)]_{y_i^G} - [f(z_i)]_j}
    {\|W_{y_i^G} - W_j\|_2}.
    \end{equation}
    If $\|W_{y_i^G}-W_j\|_2=0$, the per-pair logit difference is constant and corresponds to no finite constraint.
    The tight per-sample radius takes the minimum:
    $r_i^* = \min_{j\neq y_i^G}\, r_{i,j}$.
    Since $\|W_{y_i^G} - W_j\|_2 \le \sqrt{2}\,\sigma_{\max}(W)$ and
    $[f(z_i)]_{y_i^G} - [f(z_i)]_j \ge m_i$,
    we always have $r_{i,j} \ge r_i$, and hence $r_i^* \ge r_i$.
    For consistency of exposition, the rest of the paper mainly uses the conservative radius $r_i$; the experiments report results for both radii.
    \end{remark}

    \begin{remark}[Certificate scope and capability boundary]
    \label{remark:scope}
    The object certified by Lemma~\ref{lemma:safety} is the local consistency of the \texttt{GuardNet} classifier in the feature space: it guarantees that the \texttt{GuardNet} prediction label $y_i^G$ is invariant inside the certified ball. When the cached payload $\tilde{y}_i$ happens to equal $y_i^G$ (the proxy-label mode), the label returned online is directly aligned with the certificate---this is the default setting adopted in the main experiments of this paper.

    However, the lemma does \emph{not} provide formal guarantees for the following cases: (1) when the payload source is the \texttt{MainNet} prediction $\hat{y}_M(x_i)$ or the ground-truth label $y_i$ and it is inconsistent with $y_i^G$, the cached payload $\tilde{y}_i$ returned on a hit may differ from the certified $y_i^G$; (2) whether the \texttt{MainNet} prediction over the entire certified ball is consistent with the payload. In other words, the theoretical certificate of \texttt{LipCache} is on the \texttt{GuardNet} side, not on the \texttt{MainNet} side. The consistency between \texttt{GuardNet} and \texttt{MainNet} is indirectly promoted by the training objective and evaluated empirically in the experiments. To obtain a formal \texttt{MainNet}-side guarantee, one may further stack Lipschitz certification or randomized smoothing of \texttt{MainNet} itself.
    \end{remark}

    Combining the derivations of this section, the online cache-hit decision of \texttt{LipCache} obtains a complete theoretical skeleton: starting from the Lipschitz control of \texttt{GuardNet} (Section~\ref{subsec:guardnet_training}), Lemma~\ref{lemma:safety} turns the local classification margin into a computable per-sample certified radius, Remark~\ref{remark:tighter_radius} provides a larger and less conservative per-sample certified radius, Corollary~\ref{cor:input} extends the condition to input space, and finally Remark~\ref{remark:scope} clarifies the applicable scope of the certificate. Under this theoretical framework, the experimental evaluation in Section~\ref{sec:experiments} reports not only the hit rate and accuracy, but also verifies, via the ``certified-consistency rate,'' whether every accepted cache reuse indeed occurs inside the certified boundary.

\section{Experiments}
\label{sec:experiments}

Section~\ref{sec:theory} establishes the theoretical guarantee for cache reuse in \texttt{LipCache}: any query that falls inside the certified radius $r_i$ necessarily has a \texttt{GuardNet} prediction label consistent with the cache center. This section subjects that theoretical guarantee to empirical examination, organizing the experimental evaluation around three questions. First, on real image classification tasks, can certified caching form non-trivial reuse and stably hold under limited accuracy loss. Second, is the certified radius necessary relative to empirical thresholds---that is, does there exist a simpler alternative boundary that equally satisfies certified consistency. Third, with the certification condition held fixed, how do cache strategy, feature dimensionality, training objective, and class count jointly determine the operating point among hit rate, accuracy, and latency gains. It should be noted that the evaluation goal of this section is the viability of the certification mechanism itself under i.i.d. test distributions, i.e., examining whether the certified radius can form non-trivial reuse while preserving theoretical consistency. Temporal redundancy, concept drift, and consecutive-frame reuse rates in streaming deployment constitute evaluation problems of a different nature, and their benchmark design lies outside the scope of this paper (see Section~\ref{subsec:future_work}).

The main text consistently reports three core metrics: the cache hit rate $H$ (the frequency with which the system reuses the cache), the end-to-end accuracy (the agreement between the overall system output and the ground truth), and the certified-consistency rate (the fraction of all accepted cache hits that satisfy the certification condition; $100\%$ means that no theoretical violation occurred). In addition, the measured system-level speedup is defined as
\[
S_{\mathrm{time}}=
\frac{T_{\mathrm{Main}}}
{T_{\mathrm{guard}}+T_{\mathrm{search}}+T_{\mathrm{fallback}}},
\]
where the numerator is the measured time of executing \texttt{MainNet} alone, and the denominator is the measured total time of \texttt{GuardNet} encoding, cache search, and, on a miss, the \texttt{MainNet} fallback. Although the speedup is sometimes written as $1/(1-H)$, this expression only reflects the reduction in \texttt{MainNet} invocations and ignores the encoding and search overhead; we therefore adopt the measured definition $S_{\mathrm{time}}$ above as the reported metric.

\subsection{Experimental Setup}
\label{subsec:exp_setup}

To keep the experiments comparable, the main-text experiments share a common base protocol. \texttt{MainNet} is responsible only for fallback inference on misses and serves as the standalone-accuracy and timing baseline; \texttt{GuardNet} is responsible for extracting features that satisfy the $1$-Lipschitz constraint and accordingly defining the theoretical radius for cache reuse. Unless stated otherwise, the three 10-class tasks all use the tight certified radius, a cache-construction method that prioritizes the certified radius combined with de-duplication, returns the proxy label on a hit, caches $200$ samples per class, and repeats the experiments over random seeds $\{42,123,2024\}$. The standalone accuracies of the three \texttt{MainNet} models are as follows: $94.78\%$ on CIFAR-10~\cite{krizhevsky2009cifar}, $87.0\%$ on Tiny-ImageNet 10 classes (denoted Tiny-10)~\cite{le2015tinyimagenet}, and $97.19\%$ on SVHN~\cite{netzer2011svhn}. In the current final experiments, the certified-consistency rate of all certified-radius results is $100\%$, and no Lipschitz-constraint violation or theoretical violation at the proxy-decision level is observed.

The training-objective ablation changes only how \texttt{GuardNet} is trained, with everything else held fixed; the multi-class extension experiment changes only the class count and the corresponding \texttt{GuardNet} training recipe; the perturbation and deployment experiments examine the behavior of the system under distributional perturbation and latency constraints on top of the default operating point.

\textbf{Scope of the evaluation}: The experiments are designed to validate the viability and cross-task stability of the per-sample certified reuse mechanism, rather than to propose a new SOTA image classification benchmark. The three tasks respectively represent three typical 10-class distributions---natural images, fine-grained classification, and domain shift---which suffice to examine the sensitivity of the certified radius to task geometry. Extension to full ImageNet (1000 classes), CIFAR-100, and streaming video benchmarks lies outside the scope of this paper and is left to future work.

\subsection{Feasibility of Certified Caching}
\label{subsec:certificate_validation}

    Figure~\ref{fig:main_results} reports the hit rate, end-to-end accuracy, certified-consistency rate, and measured speedup on the three 10-class tasks. Figure~\ref{fig:main_profile_heatmap} portrays the three tasks side by side from the perspective of relative operating points.

Under the conservative certified radius, the average hit rates of CIFAR-10, Tiny-10, and SVHN are $0.341\pm0.010$, $0.303\pm0.013$, and $0.470\pm0.007$, respectively; switching to the tight certified radius, the hit rates rise to $0.342\pm0.011$, $0.472\pm0.012$, and $0.502\pm0.013$. The corresponding end-to-end accuracies remain at $0.921$, $0.867$, and $0.960$, and all results pass the certified-consistency check.

A non-trivial certified reuse region is observed on all three tasks. Under the tight certified radius, the hit rates of CIFAR-10, Tiny-10, and SVHN are all significantly above zero, indicating that \texttt{LipCache} does not trigger theoretically safe reuse only in the vicinity of a tiny minority of samples, but rather can form stable usable hits over the full test set. Meanwhile, the end-to-end accuracy is only about $2.70$, $0.27$, and $1.19$ percentage points lower than the respective standalone \texttt{MainNet} accuracy, showing that certified caching does not trade a large loss in prediction quality for reuse opportunities.

The difference between the conservative and the tight certified radius reveals the basic principle that ``the certification boundary can be tightened but cannot be loosened arbitrarily.'' On CIFAR-10, the two radii are already close, and the tight radius brings only a $0.15$ percentage-point gain in hit rate and about $0.04\times$ in speedup, with accuracy nearly unchanged. On Tiny-10, the tight certified radius raises the hit rate from $0.303$ to $0.472$ and the measured speedup from $1.11\times$ to $1.31\times$, with accuracy dropping by only about $0.20$ percentage points---the conservative radius is overly cautious on higher-resolution images. SVHN lies between the two: the hit rate rises from $0.470$ to $0.502$, the speedup from $1.58\times$ to $1.65\times$, and the accuracy drops by about $0.21$ percentage points. Overall, the tight certified radius trades a moderately enlarged certified reuse region for a higher hit rate, but without breaking certified consistency.

The differences across the three datasets reflect the systematic influence of task geometry on the efficiency of certified reuse. SVHN simultaneously exhibits the highest hit rate and the highest measured speedup---on tasks such as digit images, where intra-class variation is constrained, the local geometry learned by \texttt{GuardNet} is more stable, and the certified reuse region more easily covers test samples. The hit rate of Tiny-10 is notably higher than that of CIFAR-10, but its speed gain does not increase proportionally, because the higher input resolution raises the cost of \texttt{GuardNet} encoding and nearest-neighbor search, weakening the translation of hit-rate growth into real wall-clock gains. CIFAR-10 has the lowest hit rate, yet it still maintains a measured speedup above $1.3\times$---under a lower encoding cost, medium-scale certified hits are already sufficient to yield stable system gains.

Figure~\ref{fig:main_profile_heatmap} displays, in a column-normalized side-by-side manner, the hit rate under the tight certified radius, the hit-rate gain relative to the conservative radius, the accuracy retention relative to \texttt{MainNet}, the overall prediction agreement with \texttt{MainNet}, and the measured speedup. SVHN dominates simultaneously on hit rate, accuracy retention, and speedup; Tiny-10 stands out most on hit-rate gain; and CIFAR-10 maintains a high level of agreement with \texttt{MainNet} but has a limited hit-rate gain. This structured comparison further shows that the different benefit dimensions of \texttt{LipCache} are emphasized differently across tasks.

In terms of random-seed stability, the fluctuations of the three main results are all small. Under the tight certified radius, the hit-rate standard deviations of CIFAR-10, Tiny-10, and SVHN are $0.011$, $0.012$, and $0.013$, respectively, and the accuracy standard deviations all remain at about the $0.002$ level or below, indicating that the main conclusions do not depend on a particular training seed or cache-sample composition.

Overall, Figures~\ref{fig:main_results} and \ref{fig:main_profile_heatmap} jointly establish a consistent result pattern for \texttt{LipCache} on three 10-class tasks with markedly different statistical structure: non-zero certified hits, limited accuracy loss, measurable latency gains, and zero theoretical violations.

\begin{figure*}[!htbp]
\centering
\includegraphics[width=0.98\linewidth]{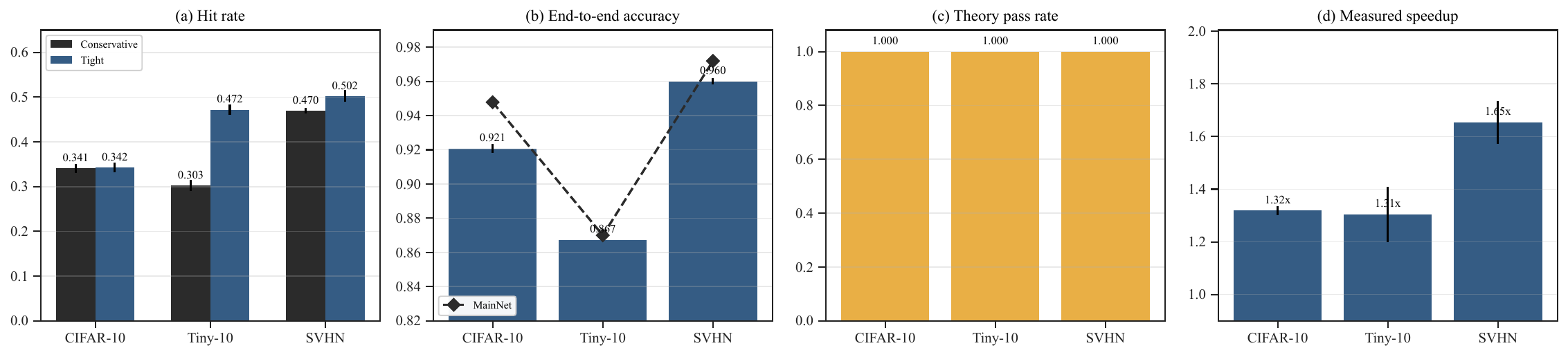}
\caption{Main results on the three 10-class tasks. From left to right: cache hit rates under the conservative versus tight certified radius; end-to-end accuracy under the tight radius, with the standalone \texttt{MainNet} accuracy marked by a dashed line; certified-consistency rate; and measured speedup. All certified-radius results pass the consistency check.}
\label{fig:main_results}
\end{figure*}

\begin{figure*}[!htbp]
\centering
\includegraphics[width=0.98\linewidth]{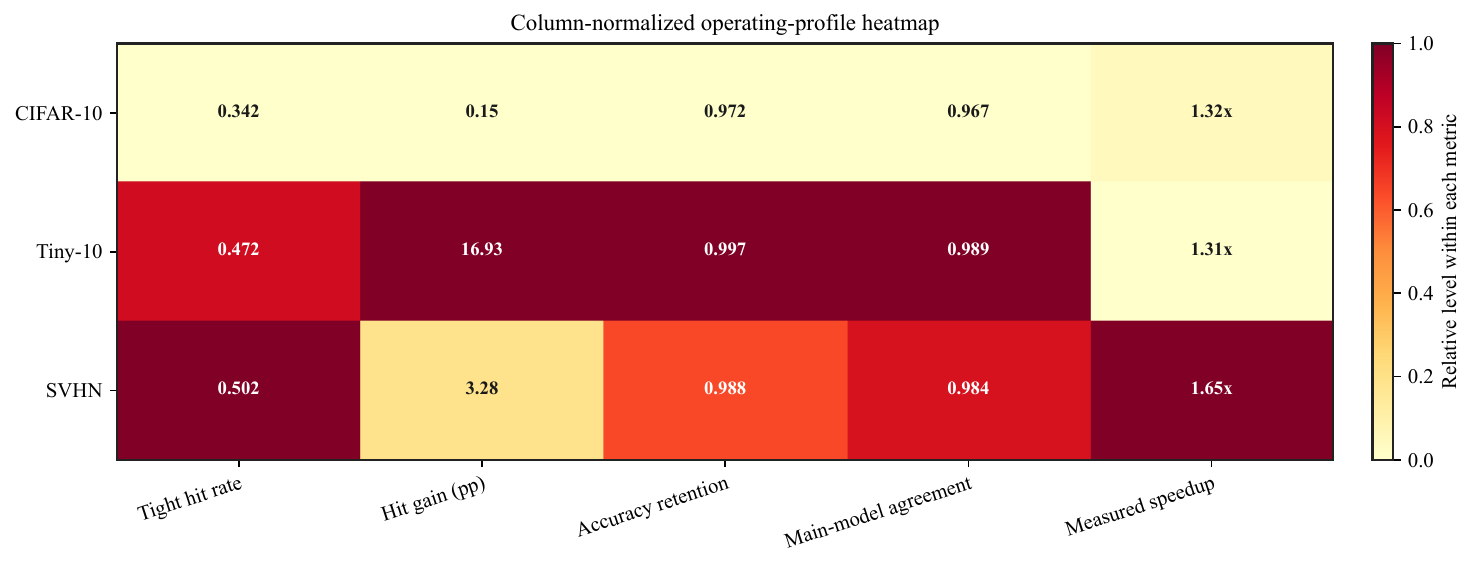}
\caption{Operating-point profile heatmap of the three 10-class tasks. The columns correspond to the hit rate under the tight certified radius, the hit-rate gain relative to the conservative radius, the accuracy retention relative to \texttt{MainNet}, the overall prediction agreement with \texttt{MainNet}, and the measured speedup. Colors are normalized within each column; the annotated numbers retain their original magnitudes.}
\label{fig:main_profile_heatmap}
\end{figure*}

\subsection{Comparison with Related Efficiency Systems}
\label{subsec:related_scope}

Related work shares the goal of reducing inference cost, but the savings occur at different points in the execution path. BranchyNet~\cite{branchynet_2016} and MSDNet~\cite{msdnet_2018} retrain multi-exit networks so that easy inputs leave early at the execution node; they reduce computation after a request reaches that node and do not reuse results associated with prior inputs. NoScope~\cite{noscope2017} specializes a cascade to fixed-camera video, DeepCache~\cite{deepcache2018} reuses regions across consecutive video frames, and Semantic Memory~\cite{li2021semanticmemory} applies empirical semantic caching through cross-layer similarity. Their locality assumptions, acceptance conditions, and deployment paths differ from ours. \texttt{LipCache} keeps \texttt{MainNet} unchanged, uses a lightweight \texttt{GuardNet} to retrieve cached proxy labels at the edge, and contacts the cloud only on a miss. It therefore first reduces uploads and cloud invocations, and can then be composed with a cloud early-exit network to reduce the cloud computation of misses.

\subsubsection{Edge--Cloud Resource Accounting}
\label{subsubsec:edge_footprint_comparison}

Table~\ref{tab:edge_footprint_comparison} reports both edge-resident model footprint and a common edge--cloud deployment accounting model. Parameter counts and raw FP32 weights describe static weights in the current CIFAR-10 implementations, excluding activations, runtime state, cache entries, and framework overhead; they are therefore necessary but not sufficient evidence of edge-deployment feasibility. If an early-exit network is placed at the edge, its complete backbone and all exits must remain resident: the three-exit BranchyNet has $1.759$M parameters/$6.71$ MiB and the seven-exit MSDNet has $2.987$M/$11.39$ MiB, whereas the base \texttt{GuardNet} ($d=64,c=32$) used by \texttt{LipCache} has $0.282$M/$1.08$ MiB. Normalize the compute budget of a complete cloud prediction to $100\%$, let the certified CIFAR-10 operating point have hit rate $H=0.348$, and use the native early-exit compute speedups relative to each system's complete backend, $S_B=2.31$ for BranchyNet and $S_M=2.14$ for MSDNet. For \texttt{LipCache}, the uploaded-request fraction is $1-H$; when a cloud early-exit backend handles only misses, normalized cloud computation is $(1-H)/S_B$ or $(1-H)/S_M$. Cloud-compute saving is $1-C$ and cloud-capacity gain is $1/C$, where $C$ is the normalized cloud computation.

\begin{table*}[!t]
\centering
\scriptsize
\renewcommand{\arraystretch}{1.15}
\caption{Edge-resident model footprint and edge--cloud resource accounting at the CIFAR-10 operating point. Footprint denotes the static FP32 weight volume of the edge-resident model, distinct from the runtime compute cost analyzed in later subsections. Uplink traffic and cloud computation are normalized to a complete cloud prediction ($100\%$); cloud computation applies early-exit speedups as conditional-compute multipliers, yielding a deployment accounting model rather than a cross-architecture comparison of FLOPs, latency, accuracy, or energy.}
\label{tab:edge_footprint_comparison}
\begin{tabular}{@{}c l c c c c c@{}}
\toprule
ID & \shortstack{Configuration\\(edge $\rightarrow$ cloud)} & \shortstack{Edge footprint\\(params/size in FP32 MiB)} & \shortstack{Uplink traffic\\(vs. all-cloud)} & \shortstack{Cloud\\compute} & \shortstack{Compute\\saved} & \shortstack{Capacity\\gain} \\
\midrule
\multicolumn{7}{l}{\textit{Cloud-only baselines}} \\
(1) & --- $\rightarrow$ \texttt{MainNet}   & ---                                          & $100.0\%$ & $100.0\%$ & $0.0\%$  & $1.00\times$ \\
(2) & --- $\rightarrow$ BranchyNet          & $1.759$M$^\dagger$ / $6.71^\dagger$           & $100.0\%$ & $43.3\%$  & $56.7\%$ & $2.31\times$ \\
(3) & --- $\rightarrow$ MSDNet              & $2.987$M$^\dagger$ / $11.39^\dagger$          & $100.0\%$ & $46.7\%$  & $53.3\%$ & $2.14\times$ \\
\midrule
\multicolumn{7}{l}{\textit{Edge \texttt{LipCache} + cloud backend}} \\
(4) & \texttt{GuardNet} $\rightarrow$ \texttt{MainNet}  & $0.282$M / $1.08$ & $65.2\%$ & $65.2\%$ & $34.8\%$ & $1.53\times$ \\
(5) & \texttt{GuardNet} $\rightarrow$ BranchyNet        & $0.282$M / $1.08$ & $65.2\%$ & $28.2\%$ & $71.8\%$ & $3.54\times$ \\
(6) & \texttt{GuardNet} $\rightarrow$ MSDNet            & $0.282$M / $1.08$ & $65.2\%$ & $30.5\%$ & $69.5\%$ & $3.28\times$ \\
\bottomrule
\multicolumn{7}{l}{\footnotesize $^\dagger$: footprint if the cloud-resident network were deployed at the edge (shown for size comparison only).} \\
\end{tabular}
\end{table*}

The accounting exposes the complementarity. Cloud BranchyNet and MSDNet reduce cloud computation to $43.3\%$ and $46.7\%$, respectively, but leave uploads at $100\%$. \texttt{LipCache} first hits at the edge for $34.8\%$ of requests, reducing both uploads and cloud invocations to $65.2\%$. In composition, cloud computation further falls to $28.2\%$ or $30.5\%$, corresponding to $3.54\times$ or $3.28\times$ cloud-capacity gain. The $H=0.348$ value is the controlled CIFAR-10 operating point of this section and differs from the $H=0.3295$ real-device path in Section~\ref{subsec:stress_and_deployment}; the two are not mixed. Early-exit networks are not restricted to cloud deployment; the table describes their resource effect when used as the cloud backend. Certification still covers only local \texttt{GuardNet} proxy-label consistency when a cached label is returned, not the early-exit network, \texttt{MainNet}, or ground-truth correctness. Combined end-to-end latency, accuracy, edge energy, and live network traffic require measurement on a unified path.

\subsection{Certified Radius and Empirical Thresholds}
\label{subsec:pareto}

Figure~\ref{fig:pareto_frontier} compares three kinds of reuse boundaries: the tight certified radius, an empirical radius based on inter-class distance statistics, and a single global nearest-neighbor threshold. The three datasets exhibit a consistent pattern: the certified radius does not pursue the highest hit rate, but it is the only boundary definition that can stably keep the certified-consistency rate at $100\%$ across all three tasks.

Taking CIFAR-10 as an example, the hit rate, accuracy, and consistency corresponding to the tight certified radius are $0.344$, $0.922$, and $100\%$, respectively; the inter-class-distance empirical radius slightly raises the hit rate to $0.352$, but the certified-consistency rate drops to $37.9\%$; the global nearest-neighbor threshold pushes the hit rate up to $0.494$, yet it drives the end-to-end accuracy down to $0.777$. On Tiny-10, the empirical thresholds are likewise more aggressive, but their corresponding certified-consistency rates are only $33.3\%$ and $23.9\%$; on SVHN, the empirical thresholds do not even achieve a higher hit rate, yet they still depress the certified-consistency rate to $45.7\%$ and $25.4\%$. The empirical thresholds either trade accuracy and consistency for a higher hit rate, or damage the certification condition without notably improving the hit rate; they do not form a better accuracy--hit-rate frontier across the three tasks.

This result defines the role of the certified radius in the main text: it does not pursue the most aggressive empirical hit rate, but provides a reproducible, interpretable, and cross-task consistent reuse boundary.

\begin{figure*}[!htbp]
\centering
\includegraphics[width=0.98\linewidth]{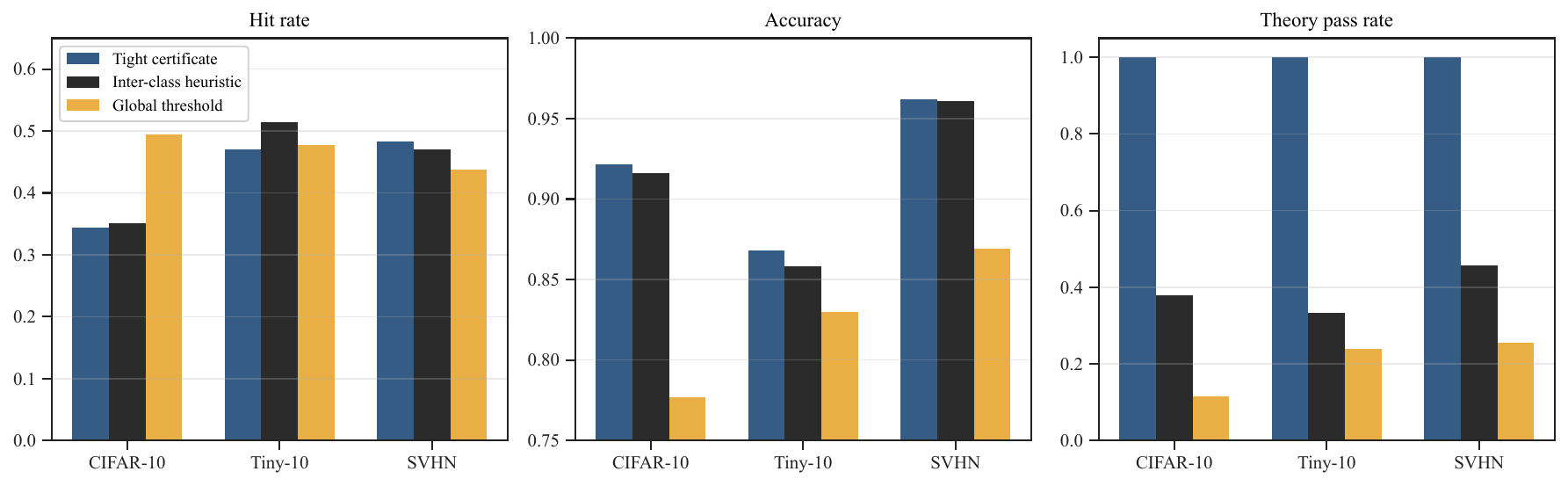}
\caption{Comparison of the certified radius and empirical thresholds. Empirical thresholds sometimes yield a higher hit rate, but only the tight certified radius keeps the certified-consistency rate at $100\%$ across all three 10-class tasks while maintaining high end-to-end accuracy.}
\label{fig:pareto_frontier}
\end{figure*}

\subsection{Default Operating Point Configuration}
\label{subsec:default_operating_point}

Figures~\ref{fig:hit_speedup_relation} and \ref{fig:dim_margin_heatmap} answer why the main-text default configuration adopts the operating point of ``tight certified radius + radius-prioritized de-duplication + proxy label + 64-dimensional features + 200 samples per class.''

On the entry-selection strategy, radius-prioritized de-duplication gives the highest hit rate on CIFAR-10 and also outperforms random sampling and K-center on Tiny-10; on SVHN, the hit rate of K-means is slightly higher than that of radius-prioritized de-duplication, but the gap is small. Considering the consistency and interpretability across the three tasks, radius-prioritized de-duplication constitutes a more robust unified default strategy, which also indicates that the certified-caching scenario requires not only geometric coverage but also prioritized retention of high-quality center samples with large certified radii.

    Regarding return-value semantics, the proxy label is directly aligned with the theoretical guarantee and provides the clearest semantics. On CIFAR-10, the end-to-end accuracies of the proxy-label, \texttt{MainNet}-label, and ground-truth schemes are $0.9215$, $0.9134$, and $0.9134$, respectively; on Tiny-10, they are $0.8680$, $0.8680$, and $0.8700$, respectively; and on SVHN, they are $0.9617$, $0.9241$, and $0.9237$, respectively. These results show that the proxy label achieves higher end-to-end accuracy on CIFAR-10 and SVHN, whereas the differences among the three semantics are small on Tiny-10. In all cases, the proxy label remains fully aligned with the certification condition. The main text adopts the proxy label as the default precisely because it has the clearest theoretical semantics under a unified protocol.

    Cache capacity remains an effective lever. On all three datasets, increasing the cache from $100$ to $400$ samples per class raises the hit rate: CIFAR-10 from $0.288$ to $0.392$, Tiny-10 from $0.446$ to $0.480$, and SVHN from $0.434$ to $0.531$. Capacity growth, however, does not automatically bring better accuracy or lower latency; the final operating point should therefore be selected jointly with the radius-tightening strategy discussed below.

The optimal feature dimensionality is dataset-dependent. On CIFAR-10, the hit rates for $d=32/64/128$ are $0.348/0.344/0.336$, respectively, and $32$ dimensions are already near-optimal; on Tiny-10, the corresponding hit rates are $0.456/0.470/0.452$, with little difference between $32$ and $64$ dimensions; on SVHN, the hit rate continues to rise from $0.498/0.484$ to $0.527$, so higher dimensions still bring gains. The main text retains $64$ dimensions as the unified default in order to share the same protocol across the three tasks; the optimal dimensionality of different datasets is not fully consistent, and this is precisely the degree of freedom that needs to be re-tuned per task in subsequent deployment.

\begin{figure*}[!htbp]
\centering
\includegraphics[width=0.98\linewidth]{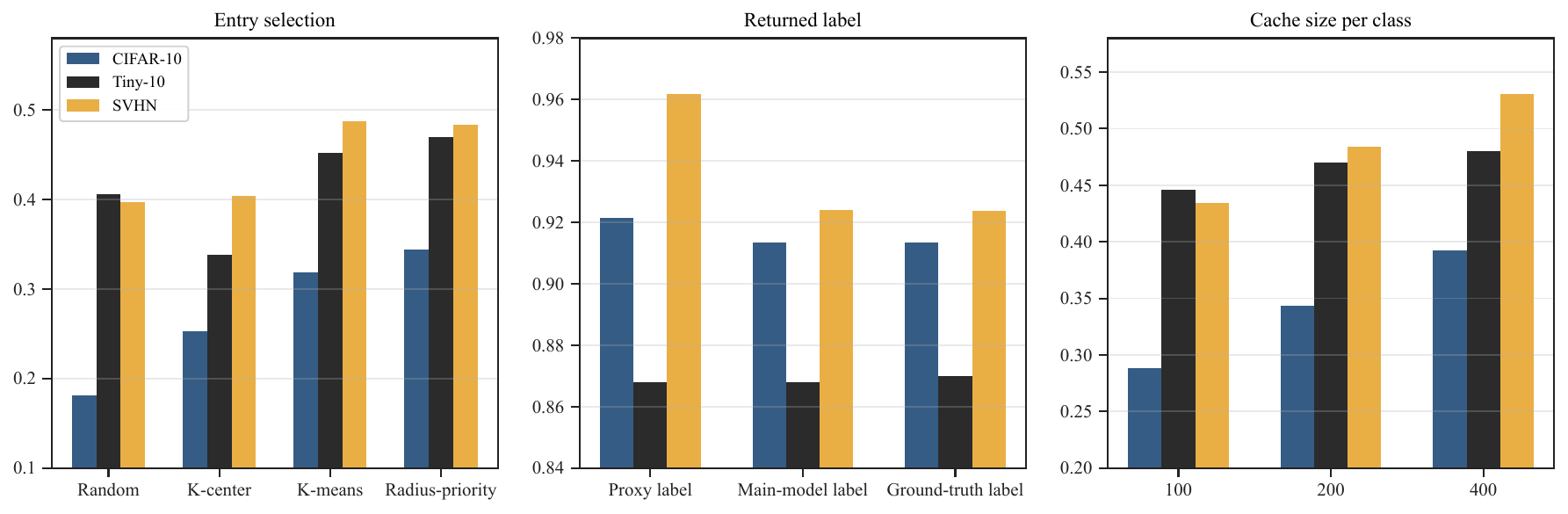}
\caption{Default-operating-point ablation on the three 10-class tasks. Left: effect of entry-selection strategy on the hit rate; middle: effect of return-value semantics on end-to-end accuracy; right: effect of per-class cache capacity on the hit rate.}
\label{fig:hit_speedup_relation}
\end{figure*}

\begin{figure*}[!htbp]
\centering
\includegraphics[width=0.9\linewidth]{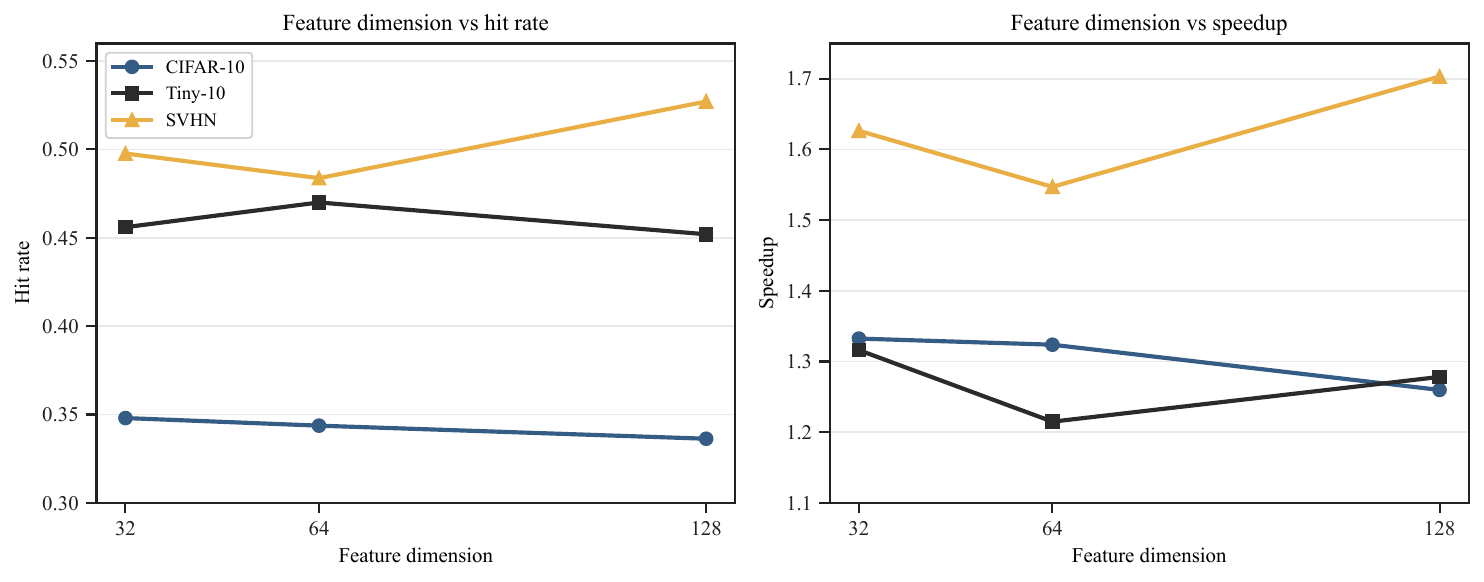}
\caption{Effect of feature dimensionality on the hit rate and the measured speedup. CIFAR-10 and Tiny-10 are already near their respective optima at smaller dimensions, whereas SVHN continues to improve at higher dimensions, indicating that the unified default dimensionality is not fully identical to each dataset's optimum.}
\label{fig:dim_margin_heatmap}
\end{figure*}

\subsection{Hit Mechanism and Radius Tightening}
\label{subsec:hit_mechanism}

After determining the default operating point, two phenomena still need to be explained: why incorrect hits occur mainly near the boundary, and why tightening the radius can improve accuracy without breaking the certification condition. Figure~\ref{fig:mechanism_shrink} provides the answer.

On all three tasks, the normalized distance ratio $\|z_q-z_i\|_2/r_i$ of correct hits is consistently and markedly lower than that of incorrect hits. The mean distance ratios of correct hits on CIFAR-10, Tiny-10, and SVHN are $0.802$, $0.752$, and $0.767$, respectively, whereas those of incorrect hits rise to $0.891$, $0.910$, and $0.908$, systematically closer to the certification boundary. Meanwhile, the cache-center classification margin of correct hits is significantly larger: $2.43$, $3.77$, and $2.38$ on the three datasets, versus only $1.63$, $1.97$, and $1.49$ for incorrect hits. This indicates that high-risk hits do not appear at random, but are concentrated in the region ``closer to the boundary and with a smaller classification margin.''

The lower panels of Figure~\ref{fig:mechanism_shrink} show the radius-tightening phenomenon directly. When the radius is progressively tightened on top of the conservative certified radius, all three datasets exhibit a decrease in hit rate and an increase in accuracy, while the certified-consistency rate remains at $100\%$ throughout. For example, on CIFAR-10 the hit rate drops from $0.344$ to $0.200$ and the accuracy rises from $0.922$ to $0.940$; on SVHN the hit rate drops from $0.460$ to $0.312$ and the accuracy rises from $0.964$ to $0.970$. Radius tightening does not change whether the certification holds; rather, by removing high-risk hits closer to the boundary, it pushes the system toward a more conservative but higher-accuracy operating point.

\begin{figure*}[!htbp]
\centering
\includegraphics[width=0.98\linewidth]{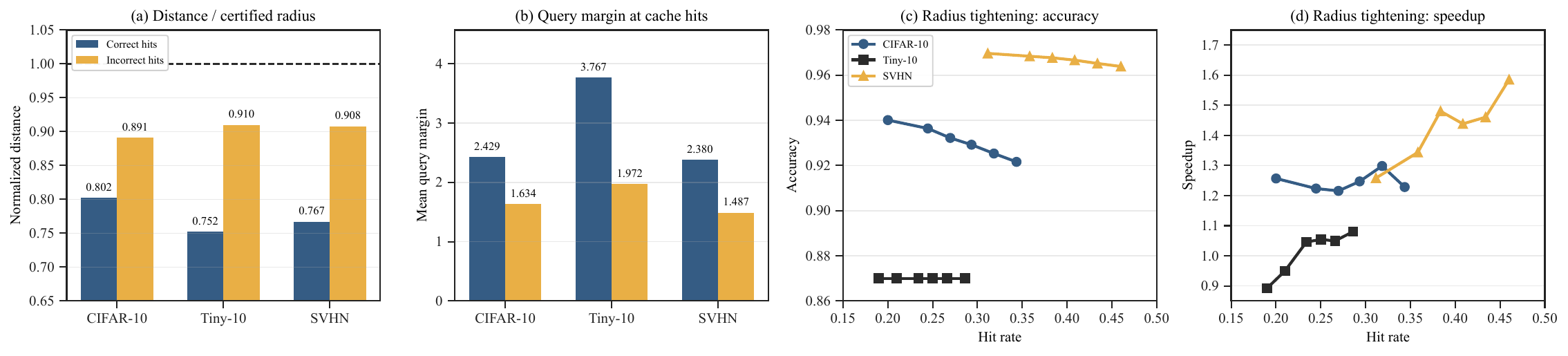}
\caption{Hit mechanism and radius tightening. Top: incorrect hits are closer to the certified boundary $\|z_q-z_i\|_2=r_i$ and have a smaller cache-hit margin. Bottom: progressively tightening the certified radius lowers the hit rate while raising accuracy and preserving theoretical self-consistency on all three 10-class tasks.}
\label{fig:mechanism_shrink}
\end{figure*}

\subsection{Training Objective and Multi-Class Extension}
\label{subsec:training_and_multiclass}

The above results show that the main degrees of freedom of the caching system lie not in whether to certify, but in how to change the hit-rate--accuracy operating point through training and radius selection. We therefore examine, respectively, the training-objective ablation on CIFAR-10 and the multi-class extension experiment on Tiny-ImageNet.

    Figure~\ref{fig:radius_mode} gives the training-objective ablation results on CIFAR-10. Cross-entropy supervision alone provides the most robust default operating point among the current 10-class main experiments, with a hit rate and accuracy of $0.358$ and $0.920$, respectively. The margin-enlargement regularizer mildly raises the hit rate: under a stronger margin-regularization setting, the hit rate rises to $0.401$ and the accuracy drops to $0.909$. The center constraint that compresses the intra-class distribution is more aggressive: the hit rate can be further raised to $0.544$, $0.596$, and even $0.644$, with the corresponding accuracy dropping to $0.894$, $0.888$, and $0.877$. It is worth emphasizing that all training-objective variants pass the certified-consistency check---the training objective changes the operating point, not the theoretical validity itself.

    Figure~\ref{fig:tiny_recipe_multiclass} answers a more critical question: after a stronger \texttt{GuardNet} training recipe, can certified caching scale from 10 classes to larger class counts? On the cross-entropy-only multi-class baseline, the hit rates under the tight certified radius at $20/30/50$ classes are only $0.056$, $0.0047$, and $0.0004$, respectively; after adopting the enhanced training recipe, these three values rise to $0.423\pm0.009$, $0.254\pm0.004$, and $0.124\pm0.001$, respectively, while the end-to-end accuracy remains at $0.838$, $0.847$, and $0.836$, and the certified-consistency rate stays at $100\%$. This indicates that the main bottleneck in the multi-class scenario lies not in the caching mechanism itself, but in whether \texttt{GuardNet} can learn a sufficiently compact and separable certified feature geometry.

    Nevertheless, the growth in class count remains the dominant factor determining the upper bound of reusability. Even with the enhanced training recipe, the hit rate keeps declining as the class count increases. Increasing the feature dimensionality cannot fundamentally reverse this trend: under the $20/30/50$-class settings, $d=32$ achieves the highest hit rates of $0.431$, $0.266$, and $0.158$, respectively, all outperforming $d=64$ under the same settings. The message conveyed by the multi-class extension is not ``one can abandon certification when the class count grows,'' but rather ``improving \texttt{GuardNet} training can markedly enlarge the usable class scale while preserving the certification condition, yet the class count still determines the upper bound of certified reuse.''

\begin{figure}[!htbp]
\centering
\includegraphics[width=0.9\columnwidth]{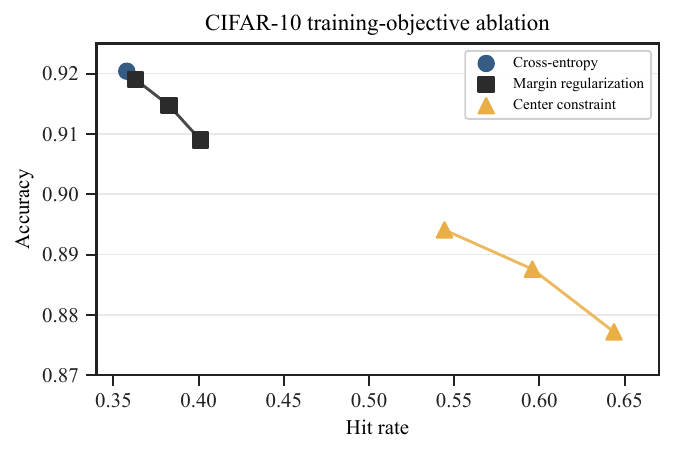}
\caption{Training-objective ablation on CIFAR-10. Cross-entropy supervision gives the most robust default operating point; margin regularization can mildly raise the hit rate; the center constraint can further raise the hit rate, but at a more noticeable accuracy cost.}
\label{fig:radius_mode}
\end{figure}

\begin{figure*}[!htbp]
\centering
\includegraphics[width=0.98\linewidth]{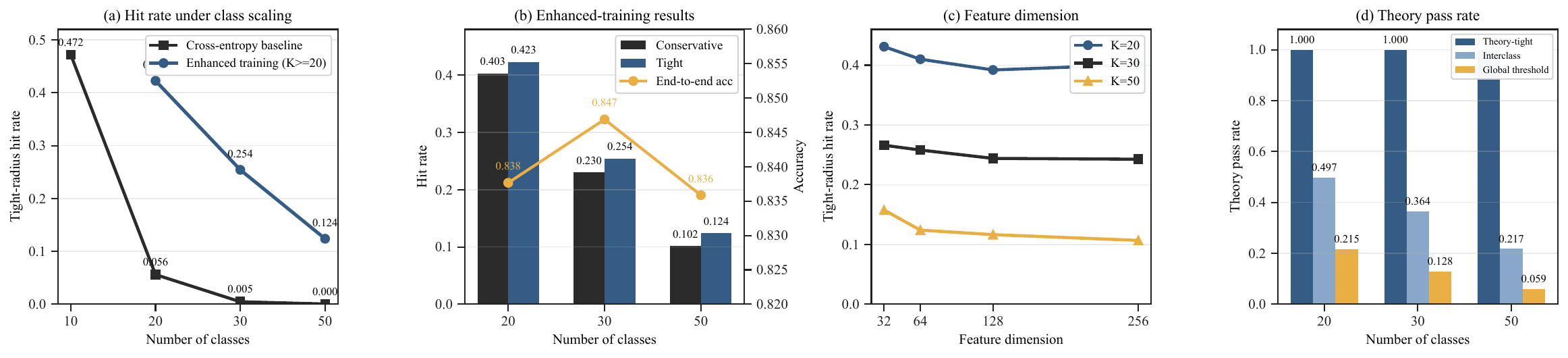}
\caption{Tiny-ImageNet multi-class extension experiment. Top left: comparison of the tight-certified-radius hit rates of the cross-entropy-only baseline and the enhanced training recipe at different class counts; top right: conservative versus tight certified-radius hit rates of the enhanced training recipe at $20/30/50$ classes and the corresponding end-to-end accuracy; bottom left: relationship between feature dimensionality and hit rate at different class scales; bottom right: certification consistency of different radius definitions in the multi-class scenario.}
\label{fig:tiny_recipe_multiclass}
\end{figure*}

\subsection{Query Perturbation and Deployment Boundaries}
\label{subsec:stress_and_deployment}

Figure~\ref{fig:knn_baseline} gives the empirical stress-test results under query perturbation. This group of experiments should be understood as an empirical generalization assessment rather than a direct extrapolation of the theoretical certificate. All three datasets exhibit the same pattern: additive noise is the most stable, JPEG compression brings a mild decline, and cropping and scaling are the most damaging to the hit rate. For example, on CIFAR-10 the baseline hit rate under the tight certified radius is $0.343$; it drops to $0.310$ under strong JPEG perturbation, and drops to $0.093$ and $0.152$ under strong cropping and strong scaling, respectively. Tiny-10 and SVHN preserve the same trend.

    More importantly, the performance degradation under perturbation comes mainly from ``not hitting'' rather than ``hitting but being wrong.'' In all perturbation experiments, the certified-consistency rate remains at $100\%$, indicating that certified consistency itself is not broken; the performance change is mainly because the perturbed query points more easily escape the original certified ball, rather than because local consistency inside the certified ball has failed.

    It should be emphasized that the above perturbation tests examine the robustness of the certified radius to single-point local perturbations and remain within the i.i.d. evaluation framework. Temporal redundancy (e.g., inter-frame coherence in video) and concept drift in streaming scenarios involve sequential correlations and require dedicated streaming benchmarks together with an analysis of how the hit rate evolves over time; these are left as future work. The results in this section can serve as an upper-bound reference for streaming deployment: if the i.i.d. hit rate is $H_0$, then under strong temporal correlation the probability that consecutive frames fall inside the same certified ball is typically no lower than $H_0$.

    We further deploy the baked CIFAR-10 \texttt{GuardNet} on an Atlas 200I DK A2 equipped with a DaVinci 310B4 NPU. The edge board executes \texttt{GuardNet} and cache lookup, while the remote \texttt{MainNet} latency is measured with ResNet50~\cite{he2016resnet} on an RTX 4070 Ti. The device evaluation uses $2{,}000$ queries and a $2{,}000$-entry cache, obtaining a hit rate of $0.3295$ and an end-to-end accuracy of $0.9170$.

    Figure~\ref{fig:edge_measurements} reports the real-device latency components and the RTT-based split-inference estimate. The NPU takes $1.34$\,ms/query for \texttt{GuardNet} encoding and the cKDTree lookup takes $0.22$\,ms/query, giving an edge-local path of $1.56$\,ms/query. Cloud batch-1 \texttt{MainNet} inference takes $8.38$\,ms/query. Combining these measured components with LAN, WiFi, and cellular RTTs of $2/10/40$\,ms gives expected speedups of $1.22\times$, $1.32\times$, and $1.42\times$, respectively.

    Figure~\ref{fig:edge_boundary} exposes two limits that are hidden by the batch-1 result. First, the cloud GPU reduces its per-image time from $8.38$\,ms at batch 1 to $0.207$\,ms at batch 128, whereas the edge NPU encoding time follows a shallow U-shape ($1.34/0.81/0.87/1.18$\,ms). Second, a discrete-event simulation driven by these measured per-operation times shows a saturation boundary: at $1{,}000$ queries/s, the edge reaches approximately $99\%$ utilization and \texttt{LipCache} p99 latency rises to $177.4$\,ms, compared with $27.8$\,ms for all-cloud inference. Thus, edge caching is advantageous when RTT is material and the edge retains spare capacity, but it can move the bottleneck from the cloud to the edge under extreme load.

    The measurement scope is deliberately separated in the figures. NPU encoding, cache lookup, and cloud inference are hardware measurements; the network speedups combine those measurements with modeled RTT rather than a live HTTP round trip; and the load experiment is a discrete-event simulation parameterized by measured operation times rather than a real traffic trace.

\begin{figure*}[!htbp]
\centering
\includegraphics[width=\linewidth]{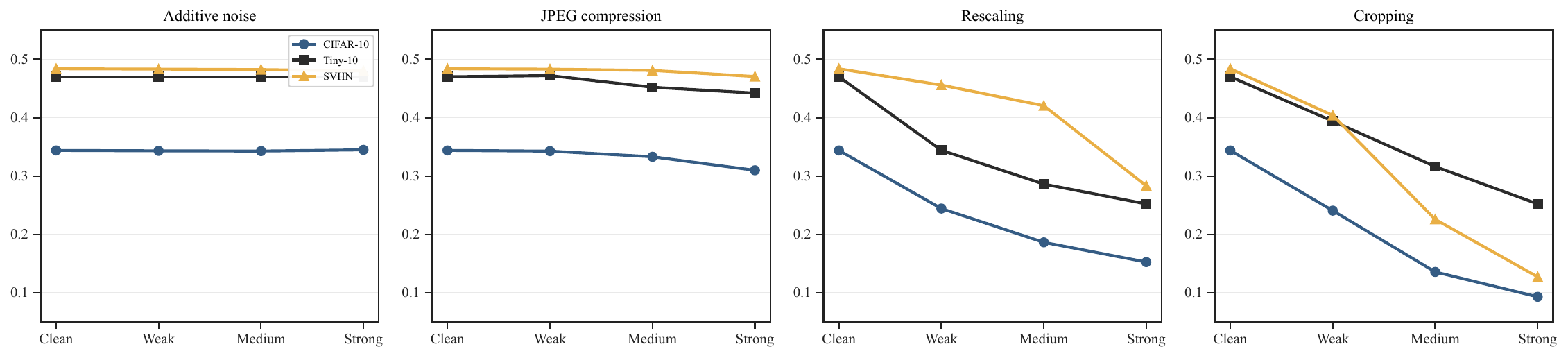}
\caption{Empirical stress test under query perturbation. Noise perturbation barely affects the certified hit rate of the three 10-class tasks; JPEG compression brings a mild decline; cropping and scaling more easily push query points out of the certified ball, so the hit-rate decline is more pronounced.}
\label{fig:knn_baseline}
\end{figure*}

\begin{figure}[!htbp]
\centering
\includegraphics[width=0.9\columnwidth]{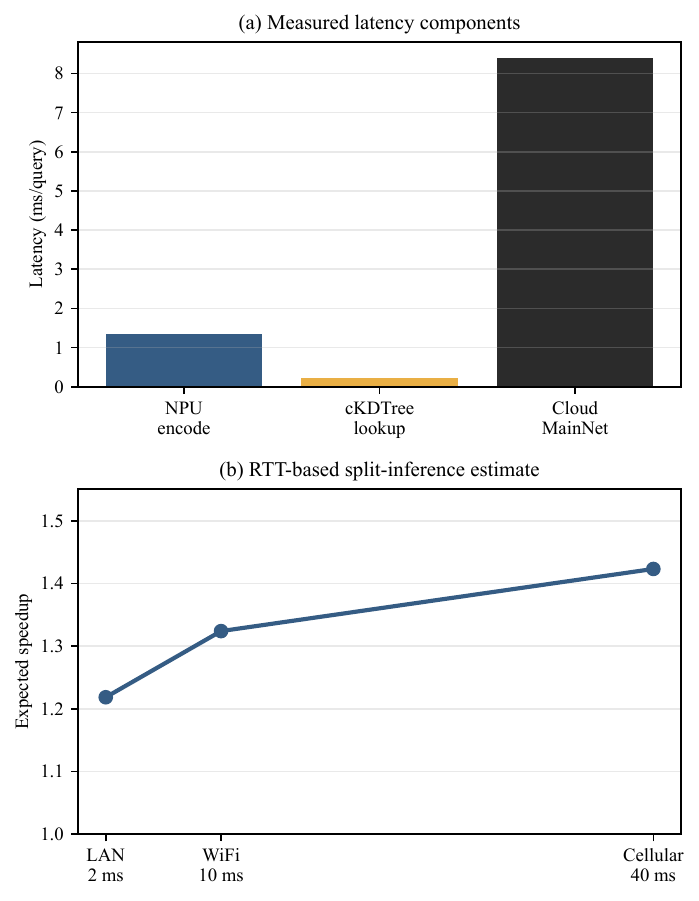}
\caption{Real-device measurements for the CIFAR-10 operating point. Left: measured latency components on the DaVinci 310B4 edge NPU and RTX 4070 Ti cloud GPU. Right: expected speedup obtained by combining the measured components with modeled LAN, WiFi, and cellular RTTs. The right panel is an analytical estimate, not a measured HTTP round trip.}
\label{fig:edge_measurements}
\end{figure}

\begin{figure}[!htbp]
\centering
\includegraphics[width=0.9\columnwidth]{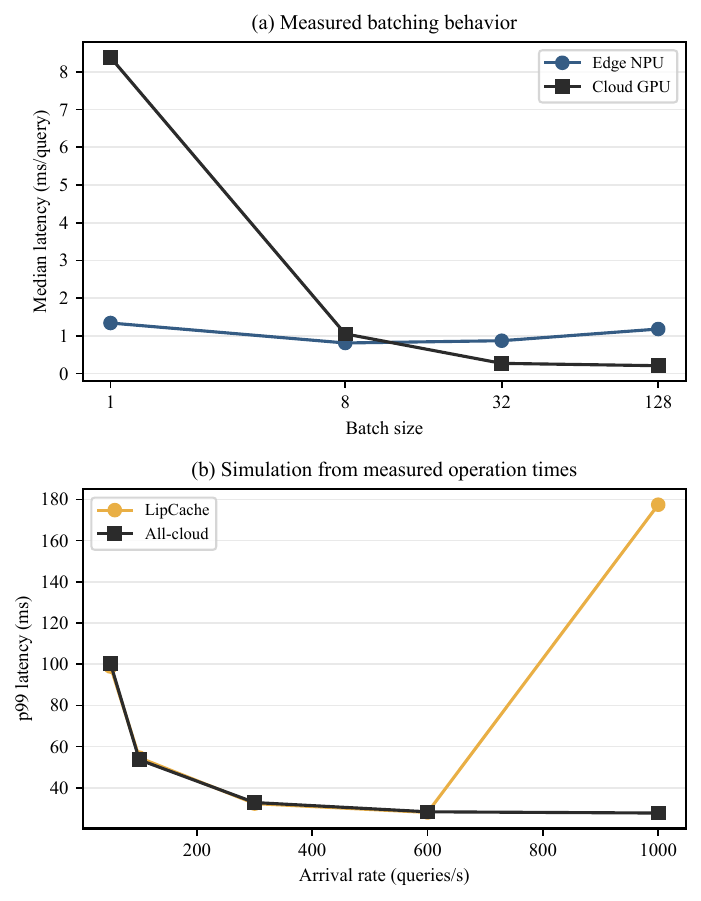}
\caption{Operating boundary of the edge deployment. Left: measured per-query latency under different batch sizes, showing the much stronger batching benefit of the cloud GPU. Right: p99 latency from a discrete-event simulation parameterized by measured operation times; the sharp rise near $1{,}000$ queries/s marks edge saturation.}
\label{fig:edge_boundary}
\end{figure}

\subsection{Retrieval Complexity Analysis}
\label{subsec:retrieval_scaling}

Figure~\ref{fig:complexity_analysis} isolates the variable retrieval term in Eq.~(16). We fix $C_G$ and $\mathrm{InvokeRate}_{M}C_M$ at the final CIFAR-10 seed-42 operating point ($d=64$, $N=2{,}000$, hit rate $0.3437$, hence $\mathrm{InvokeRate}_{M}=0.6563$). The left panel fixes $N=2{,}000$ and varies $d$; the right fixes $d=64$ and varies $N$. The unchanged additive terms are normalized to one work unit, so both panels directly show the linear contribution of $Nd$. This is an analytical complexity analysis of the online decision rule, not a hardware-latency measurement.

\begin{figure}[!htbp]
\centering
\includegraphics[width=0.9\columnwidth]{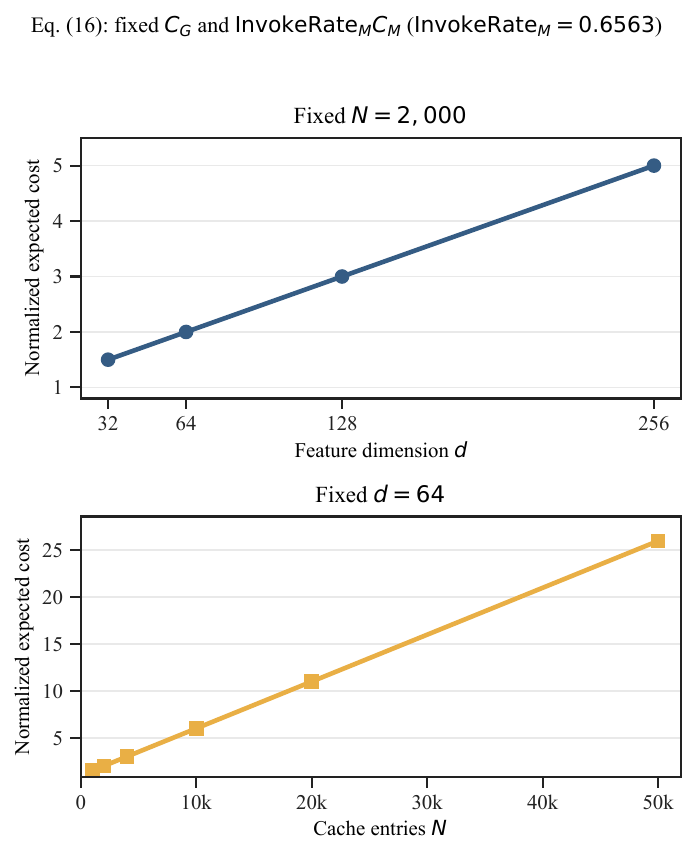}
\caption{Complexity analysis corresponding to Eq.~(16). $C_G$ and $\mathrm{InvokeRate}_{M}C_M$ are fixed at the final CIFAR-10 seed-42 operating point; the unchanged additive terms are normalized to one work unit. Left: with $N=2{,}000$ fixed, the expected cost grows linearly with $d$. Right: with $d=64$ fixed, it grows linearly with $N$.}
\label{fig:complexity_analysis}
\end{figure}

\subsection{Section Conclusions}

    The experiments in this section support the following conclusions. First, \texttt{LipCache} can stably achieve certified cache reuse on the three 10-class tasks of CIFAR-10, Tiny-10, and SVHN, and the main results consistently pass the certified-consistency check. Second, although the certified radius is more conservative than empirical thresholds, it is the only reuse boundary that can stably preserve certified consistency across tasks. Third, cache strategy, return-value semantics, feature dimensionality, and radius tightening mainly determine the operating point among hit rate, accuracy, and latency, rather than whether certification holds. Fourth, the training-objective ablation on CIFAR-10 and the multi-class extension on Tiny-ImageNet jointly show that improving \texttt{GuardNet} training can markedly change the operating point and even enlarge the usable class scale, yet the growth in class count remains the core factor limiting the upper bound of certified reuse. Finally, the perturbation, deployment, and complexity analyses show that the applicability domain of \texttt{LipCache} is jointly determined by the local feature geometry, cache scale, the edge encoding cost, and the network RTT, providing clear boundary conditions for subsequent system-level extensions.

    \section{Conclusion and Future Work}
    \label{sec:conclusion}

    This paper proposed \texttt{LipCache}, a certified semantic caching framework for resource-constrained edge image classification. \texttt{LipCache} retains the high-accuracy \texttt{MainNet} as a fallback predictor and introduces a lightweight, Lipschitz-controlled \texttt{GuardNet} that makes cache-hit decisions in a compact feature space. By deriving a per-sample certified reuse radius from the local \texttt{GuardNet} margin and the Lipschitz bound of the classification head, the framework turns cache lookup into an interpretable local-consistency check, rather than an exact-match or global heuristic-threshold decision.

    The experimental results show that, on three 10-class tasks with markedly different statistical structure, \texttt{LipCache} consistently exhibits the following pattern of results: ``non-zero certified hits, limited accuracy loss, measurable latency gains, and zero theoretical violations.'' The comparison between the certified radius and empirical thresholds further shows that the per-sample safety criterion does not pursue the most aggressive hit rate, but provides the only reuse boundary that can stably preserve certified consistency across tasks. Entry-selection strategy, return-value semantics, feature dimensionality, and radius tightening mainly affect the operating point among hit rate, accuracy, and latency, rather than whether certification holds. The multi-class extension experiment further reveals that the main limitation of the current method lies not in the cache criterion itself, but in whether \texttt{GuardNet} can learn a sufficiently compact and separable certified feature geometry. The edge-deployment analysis further indicates that the system gain of \texttt{LipCache} is jointly determined by the hit rate, the guard-side encoding cost, and the network RTT, and that it is best suited to deployment scenarios where ``cloud invocations are expensive while the edge side can afford a lightweight proxy.''

    In summary, \texttt{LipCache} offers a feasible path toward reliable cache-assisted inference at the edge, supported by theoretical analysis and experiments across multiple tasks. The remainder of this section discusses a number of open questions and extension directions along this path.

    \subsection{Future Work}
    \label{subsec:future_work}

    Before outlining future work, we first clarify the scope and boundaries of the current study, so that subsequent research can advance in a targeted manner.

    \textbf{Task scale and type.} The experiments in this paper validate the viability of the certification mechanism on three 10-class image classification tasks, designed to isolate the sensitivity of the certified radius to task geometry. Extension to larger-scale benchmarks such as full ImageNet (1000 classes) and CIFAR-100 is a natural next step.

    \textbf{Evaluation distribution.} All experiments are conducted on i.i.d. test sets, targeting a controlled validation of the certification mechanism itself. Temporal redundancy, concept drift, and consecutive-frame reuse rates in streaming deployment involve sequential correlations and require dedicated benchmarks together with a temporal analysis of the hit rate. The perturbation experiments in Section~\ref{subsec:stress_and_deployment} can serve as an upper-bound reference for streaming deployment, but a full temporal validation is left as future work.

    \textbf{Energy boundary.} The goal of this paper is cache-reuse efficiency and inference acceleration, rather than end-to-end energy reduction. Reducing the number of \texttt{MainNet} invocations mechanically reduces energy consumption, but a complete energy assessment would require hardware-in-the-loop measurement, which lies outside the methodological scope of this paper.

    Given these boundaries, although \texttt{LipCache} already achieves a good balance among accuracy, latency, and throughput, the current design still centers on cache-reuse efficiency and inference acceleration, rather than end-to-end energy reduction. For resource-constrained edge deployment, reducing the number of main-model invocations is only part of the problem. A natural next step is therefore to introduce a spiking neural network (SNN) variant of \texttt{GuardNet} on the guard side, so that the system can further reduce energy consumption while preserving discriminative performance and certified cache-reuse validity. This subsection first reports an already-implemented exploratory SNN prototype and its preliminary observations, and then, building on it, discusses broader open directions.

    \subsection{Exploration: An SNN-Based \texttt{GuardNet}}
    \label{subsec:snn_guardnet}

    The above discussion is not purely speculative. An SNN-based \texttt{GuardNet} variant has been implemented to evaluate its structural compatibility with the \texttt{LipCache} pipeline. The current design retains the same teacher--student training philosophy as the ANN \texttt{GuardNet}, but replaces static activations with multi-step spiking dynamics, so that the guard model may benefit from event-driven computation in future low-power deployments.

    Figure~\ref{fig:snn_guard_architecture}~summarizes the architecture of the current SNN-based \texttt{GuardNet}. Each stage follows a ``convolution--batch normalization--LIF neuron--average pooling'' pattern; higher-layer feature maps are further projected through linear feature layers and additional LIF units, then averaged over time and mapped back into a bounded feature space. The implementation explicitly exposes the number of time steps $T$, the membrane constant $\tau$, the input scaling, and an optional Poisson-style input encoding, making the model suitable for further analysis of temporal-encoding behavior.

    \begin{figure}[!htbp]
    \centering
    \includegraphics[width=0.98\linewidth]{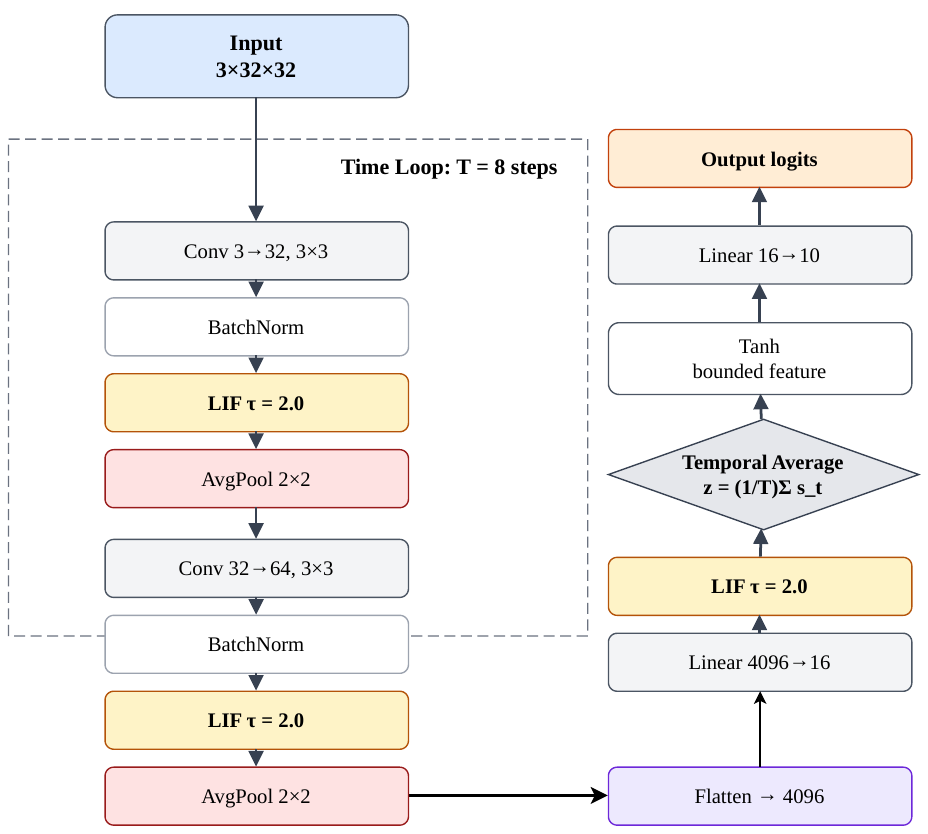}
    \caption{Architecture of the current SNN-based \texttt{GuardNet}. The model replaces the static activations of the original ANN \texttt{GuardNet} with multi-step spiking dynamics, while still producing compact features for \texttt{GuardNet}-side classification and future cache lookup.}
    \label{fig:snn_guard_architecture}
    \end{figure}

    Preliminary observations on CIFAR-10 show that the SNN-based \texttt{GuardNet} is trainable and inherits part of the semantic structure learned by the ANN teacher, but they also clearly indicate that the current version is not yet sufficient to replace the ANN guard in the main \texttt{LipCache} pipeline. Specifically, the SNN-based \texttt{GuardNet} reaches a classification accuracy of $74.26\%$ and an agreement rate of $74.78\%$ with \texttt{MainNet}, whereas under the same feature dimensionality the ANN \texttt{GuardNet} reaches $76.84\%$ accuracy and $77.13\%$ agreement. In the current software implementation, the SNN-based \texttt{GuardNet} is also slower---the measured per-sample latency is $6.97$~ms versus $0.50$~ms for the ANN \texttt{GuardNet}, although the two checkpoints are comparable in size ($338.8$~KB vs.\ $355.9$~KB). Figure~\ref{fig:snn_future_comparison}~summarizes this preliminary comparison.

    \begin{figure*}[!htbp]
    \centering
    \includegraphics[width=0.98\linewidth,height=2in]{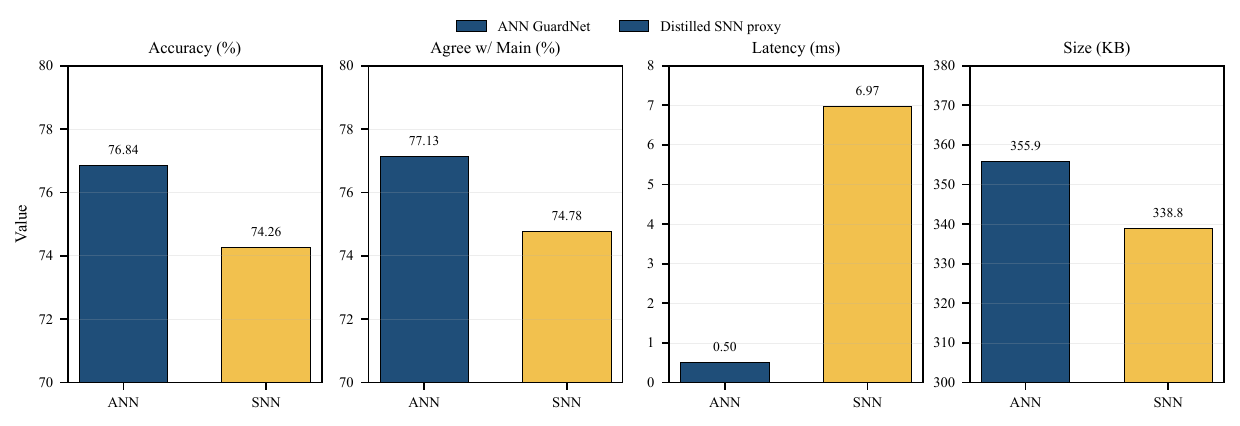}
    \caption{Preliminary comparison on CIFAR-10 between the ANN-based \texttt{GuardNet} and the currently distilled SNN-based \texttt{GuardNet}: (a) classification accuracy, (b) consistency with \texttt{MainNet}, (c) per-sample latency, (d) checkpoint size.}
    \label{fig:snn_future_comparison}
    \end{figure*}

    These results suggest a balanced reading. On the positive side, they indicate that the \texttt{GuardNet} learning framework of \texttt{LipCache} is compatible with more than one model family: the guard does not necessarily have to remain an ANN. At the same time, the current SNN study should be understood as exploratory evidence rather than a deployment-level conclusion---the SNN-based \texttt{GuardNet} has not yet reached the predictive quality required for reliable cache hits, and its measured latency was obtained on a conventional PyTorch/SpikingJelly software stack rather than on event-driven hardware. The current evidence therefore supports feasibility, not superiority.

    Based on these preliminary observations, we outline below a number of open questions in future work.

    First, energy-efficiency evaluation requires a unified and physically meaningful accounting rule. Directly comparing a heuristic spike-count estimate of an SNN with a FLOP-style approximation of an ANN is insufficient. Both models should be analyzed under an operation-level formulation---ANN inference is measured in multiply--accumulate operations, and event-driven SNN inference is measured in accumulate operations~\cite{lee2020enabling,horowitz2014computing}.

    Second, the SNN-based \texttt{GuardNet} itself needs significant improvement before it can be integrated as a reliable guard. Promising directions include stronger distillation objectives, better temporal-encoding schemes, explicit margin-aware training tailored to spiking features, and meaningful architectural constraints capable of recovering the Lipschitz control used by the ANN \texttt{GuardNet}.

    Third, the resulting energy and performance claims should be validated on neuromorphic hardware, including Loihi, TrueNorth, or SpiNNaker-class systems~\cite{davies2018loihi,merolla2014million,matsuo2024neuromorphic}. Such studies should go beyond coarse spike statistics and explicitly account for synaptic events, routing overhead, memory access, and board-level power measurement.

    Finally, an improved SNN-based \texttt{GuardNet} should be integrated into the complete \texttt{LipCache} inference pipeline, so that its impact on hit rate, fallback rate, end-to-end latency, and overall throughput can be evaluated at the full-system level rather than only at the standalone level.

    Taken together, these directions point to a broader research path: \texttt{LipCache} evolving from a cache-assisted acceleration framework into an energy-aware edge inference system. If \texttt{GuardNet} modeling, cache design, and hardware deployment can be jointly optimized, \texttt{LipCache} may ultimately become a unified solution that balances accuracy, latency, throughput, and energy.

    \section{Acknowledgment}
    Z. Xiang conceived the \texttt{LipCache} idea and developed the theoretical proofs. Y. Chen and F. Ying conducted the experiments, produced the figures, H. Zhao contributed to the framework design. Z. Xiang and Y. Chen proposed the SNN-based \texttt{GuardNet}, and Y. Chen carried out the preliminary experiments. B. Zhou and S. Dustdar played an important role in problem modeling and feasibility validation. The manuscript was polished with the assistance of generative-AI tools (GLM~5.2\footnote{https://www.bigmodel.cn} and DeepSeek~V4\footnote{https://chat.deepseek.cn}).

    \ifCLASSOPTIONcaptionsoff
      \newpage
    \fi

    \bibliographystyle{IEEEtran}
    \bibliography{ref}
    
\newpage

    \begin{IEEEbiography}[{\includegraphics[width=1in,height=1.25in,clip,keepaspectratio]{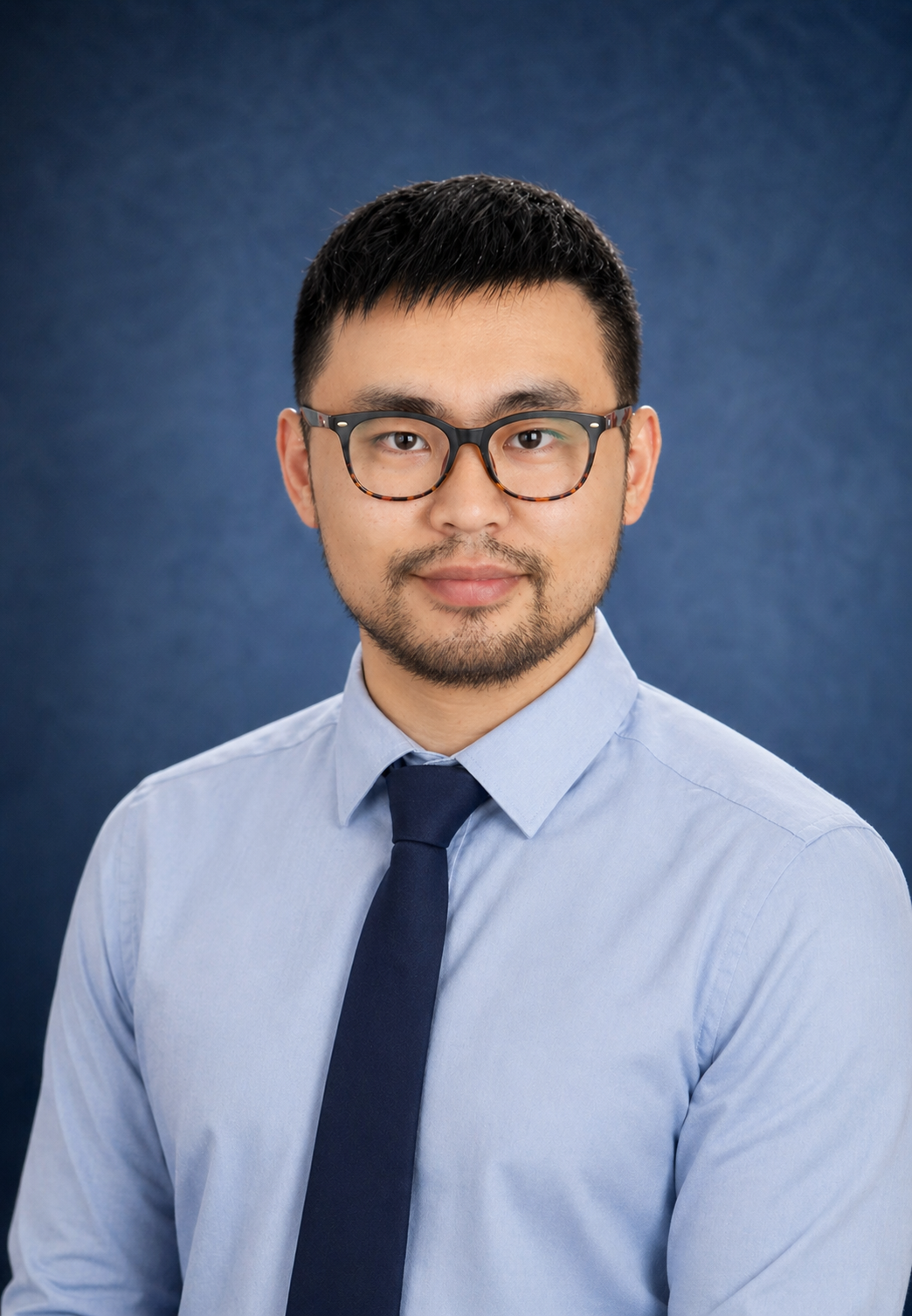}}]{Zhengzhe Xiang} received the B.S. and Ph.D. degrees in computer science and technology from Zhejiang University, Hangzhou, China. He was a Visiting Scholar with Shanghai Jiao Tong University, Shanghai, China, in 2022. He is currently an Associate Professor with Hangzhou City University, Hangzhou, China. His research interests focus on service computing and edge computing. He serves as a reviewer for several international journals, such as IEEE Transactions on Mobile Computing, IEEE Transactions on Services Computing, IET Communications, and Digital Communications and Networks. He also serves as a program committee member for many international conferences.
    \end{IEEEbiography}

    \begin{IEEEbiography}[{\includegraphics[width=1in,height=1.25in,clip,keepaspectratio]{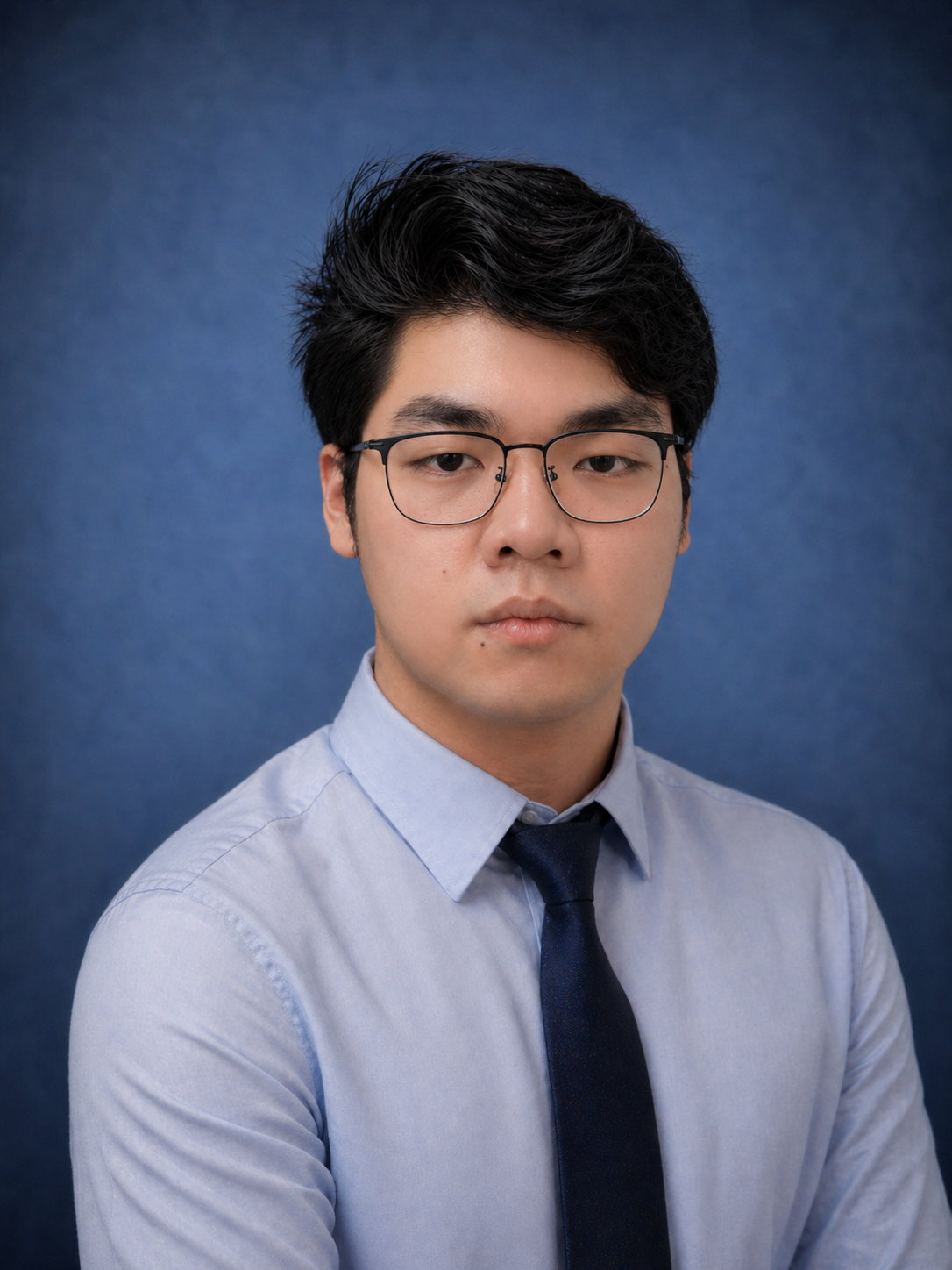}}]{Yinlin Chen} is currently pursuing the M.S. degree with the School of Computer and Computational Sciences, Hangzhou City University, Hangzhou, China. His research interests include edge computing, large language model applications, and distributed intelligent systems, with a focus on application optimization, task offloading, and intelligent resource scheduling for large language models in resource-constrained edge computing environments.
    \end{IEEEbiography}

    \begin{IEEEbiography}[{\includegraphics[width=1in,height=1.25in,clip,keepaspectratio]{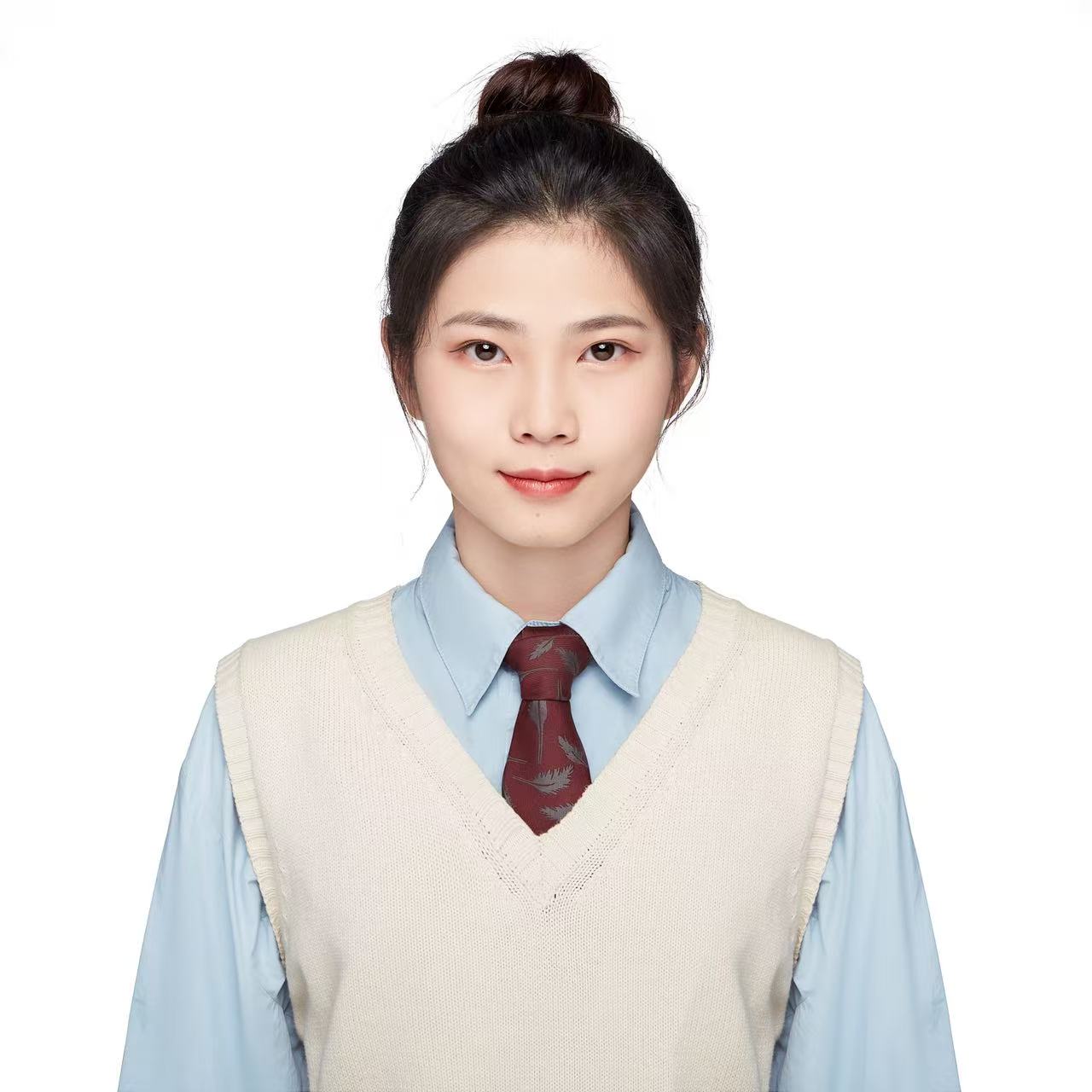}}]{Fuli Ying} is currently pursuing the M.S. degree with the School of Computer and Computational Sciences, Hangzhou City University, Hangzhou, China. Her research interests span mobile edge computing, deep-learning architectures, and image processing, with a focus on developing efficient algorithms for distributed computing environments and intelligent visual analytics systems.
    \end{IEEEbiography}

\begin{IEEEbiography}[{\includegraphics[width=1in,height=1.25in,clip,keepaspectratio]{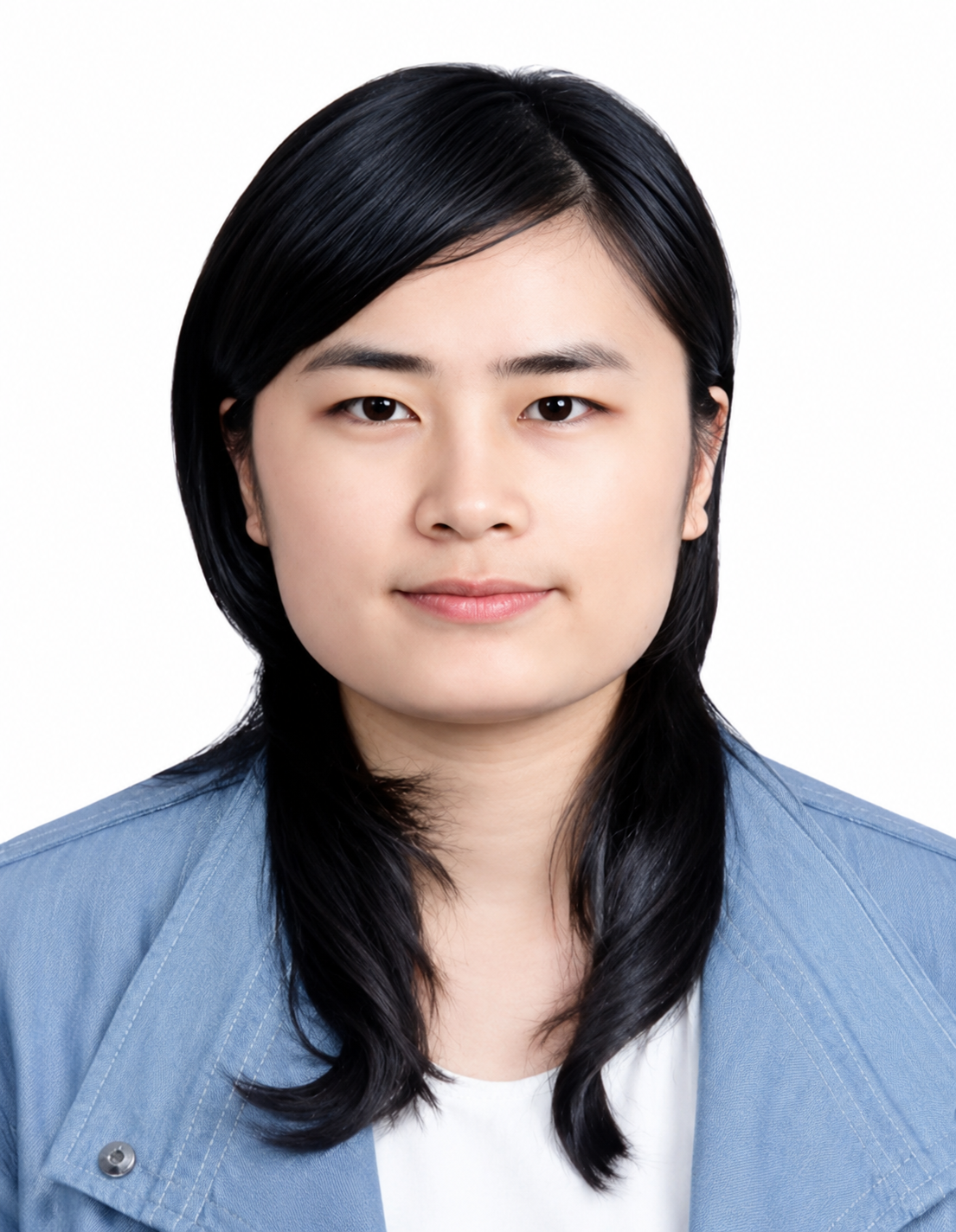}}]{Binbin Zhou} received the Ph.D. degree in computer science from Zhejiang University, Hangzhou, China, in 2021. She is currently an Associate Professor with the School of Computer Science and Computing, Hangzhou City University, Hangzhou, China. Her main research areas include spatiotemporal deep learning, multimodal fusion, artificial intelligence, and brain-inspired computing.
    \end{IEEEbiography}

    \begin{IEEEbiography}[{\includegraphics[width=1in,height=1.25in,clip,keepaspectratio]{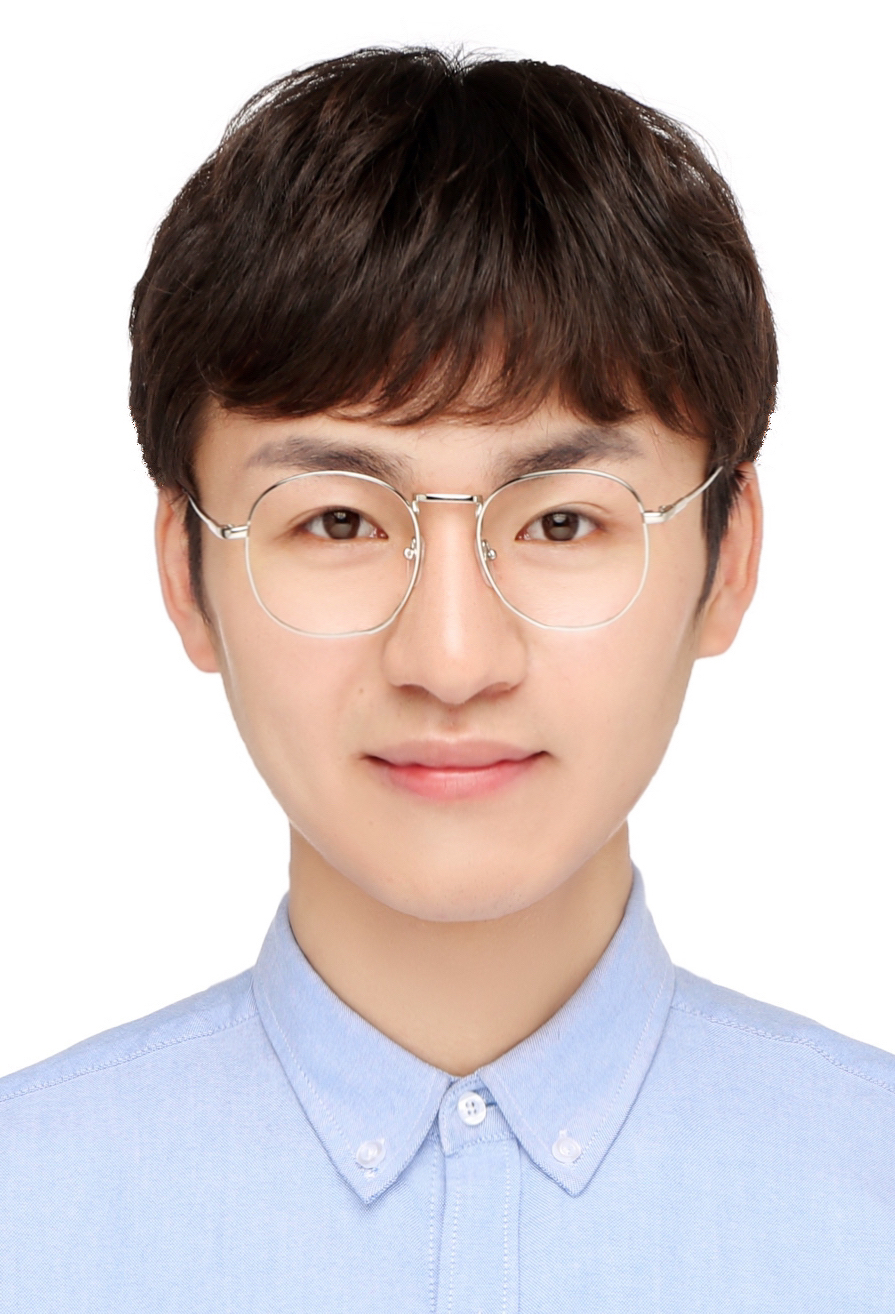}}]{Hailiang Zhao}(Member, IEEE) received the PhD degree in Computer Science and Technology in 2024. He is currently an assistant professor with the School of Software Technology, Zhejiang University, China. His research interests include distributed computing and services computing. He has published several papers in flagship conferences and journals such as IEEE ICWS, IEEE Transactions on Parallel and Distributed Systems, IEEE Transactions on Mobile Computing, IEEE Transactions on Services Computing, Proc. IEEE, etc. He was a recipient of the Best Student Paper Award of IEEE ICWS 2019.
    \end{IEEEbiography}

    \begin{IEEEbiography}[{\includegraphics[width=1in,height=1.25in,clip,keepaspectratio]{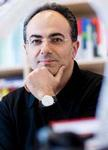}}]{Schahram Dustdar} is a Full Professor of Computer Science and heads the Distributed Systems Group at TU Wien, Vienna, Austria. He is an ICREA Research Professor at UPF Barcelona, Spain. His research interests include distributed systems, edge intelligence, and complex and autonomous software systems. He is the Editor-in-Chief of Computing; an Associate Editor of ACM Transactions on the Web, ACM Transactions on Internet Technology, IEEE Transactions on Cloud Computing, and IEEE Transactions on Services Computing. He also serves on the editorial boards of IEEE Internet Computing and IEEE Computer. He received the ACM Distinguished Scientist Award, the Distinguished Speaker Award, and the IBM Faculty Award. He is an elected member of Academia Europaea and served as the Chair of its Informatics Section from 2015 to 2022. He is an IEEE Fellow and an AAIA Fellow, and currently serves as the President of AAIA.
    \end{IEEEbiography}

    \end{document}